%% file: root.tex
\documentclass[twocolumn]{autart}

\usepackage{graphicx}  
\makeatletter
\let\c@author\relax
\makeatother
\usepackage{amsmath}
\usepackage{amssymb}
\usepackage{mathtools}
\usepackage{amsfonts}

\usepackage{physics}
\usepackage{bm}
\usepackage{xcolor}
\usepackage{wrapfig}
\usepackage[hidelinks]{hyperref}
\usepackage{upgreek}

\makeatletter
\newsavebox{\@brx}
\newcommand{\llangle}[1][]{\savebox{\@brx}{\(\m@th{#1\langle}\)}%
  \mathopen{\copy\@brx\kern-0.5\wd\@brx\usebox{\@brx}}}
\newcommand{\rrangle}[1][]{\savebox{\@brx}{\(\m@th{#1\rangle}\)}%
  \mathclose{\copy\@brx\kern-0.5\wd\@brx\usebox{\@brx}}}
\makeatother

\usepackage[backend=biber,style=authoryear,sorting=nyt,dashed=false, url=false, isbn=false, maxbibnames=5,maxcitenames=5, natbib=true]{biblatex}
\DeclareDelimFormat{finalnamedelim}{\ifnumgreater{\value{liststop}}{2}{,\space} \space\&\space}

\begin{document}

\begin{frontmatter}

\title{On the Modulating Function Method for Control Problems} 

\thanks[footnoteinfo]{The material in this paper was not presented at any conferences.}

\author[SDU]{Davi G. Accioli}\ead{davi@sdu.dk},  
\author[SDU]{Jerome Jouffroy}\ead{jerome@sdu.dk}

\address[SDU]{SDU Mechatronics, University of Southern Denmark (SDU), Alsion 2, 6400, S{\o}nderborg, Denmark}
          

\begin{abstract}
The modulating function method is an algebraic framework that, thus far, has been used for state and parameter estimation, as well as fault detection, of linear, fractional-order, distributed, and some nonlinear systems. At the core of the method lies the modulating function, which can either be selected directly or be obtained as a solution to an auxiliary system. By introducing the notion of dual modulating functions and dual modulations using auxiliary systems and duality, this paper shows that this framework is not only an estimation framework, but also a controller design framework for LTV systems. In particular, necessary and sufficient conditions for the existence of the associated control laws are introduced; the well-known state feedback law is obtained as a particular case of the dual modulation approach, along with output feedback, LTI sliding mode control, the reachability gramian, and the state-transition matrix; and a new fixed-time control law is proposed for both LTI and LTV systems, including an estimate of the transient behavior. Moreover, numerical simulations of the newly proposed control law are performed, indicating similar performance levels to a benchmark LQR even when handling unmatched disturbances.
\end{abstract}

\end{frontmatter}

\input{01_Introduction}
\input{02_Preliminaries}
\input{03_Asymptotic_Control}
\input{04_Fixed_Time_Control}
\input{05_Numerical_Examples}
\input{06_Conclusion}

\printbibliography

\end{document}

%% file: 01_Introduction.tex
\section{Introduction}\label{sec: introduction}

Control theory developments are historically linked to integral transform methods, in particular because they allow for rephrasing differential equations into algebraic ones. Although this is well-known for linear control problems due to the ubiquitous Laplace transform and PID control \citep{Astrom_PID}, integral transform methods also play a crucial role in different estimation tasks for linear and even some nonlinear systems.

The modulating function method (MFM) is an algebraic estimation framework initially proposed in \cite{Shin54} as a method for parameter estimation tasks. Since then, it has been extended to state estimation and fault detection of linear and some nonlinear systems \citep{LiuKPG14, JouffroyR15, Fischer21, Noack_Polynomial_Estimation}, as well as different estimation tasks for PDEs, fractional-order, and distributed systems \citep{Folke_PDE_Advection_Diffusion_Reaction, Zhang_MFM_LTV_2024, Kirati_Fractional_PDE18}.

The main idea behind the method consists of multiplying the differential equation by a user-defined function, called the modulating function (MF), integrating both sides of the resulting equation, and applying integration by parts: a process that can be elegantly formulated using linear operators \citep{Ungarala_MFM_Bioprocesses}. By then selecting this function such that it, and its derivatives, satisfy a particular set of boundary conditions, different estimation problems can be tackled. In particular, MFs can be divided into three classes: total modulating functions, if both left and right boundaries are zero up to a certain derivative-order; left modulating functions, if their left boundary is zero while the right boundary is not identically zero; and right modulating functions, if their right boundary is zero while the left boundary is not identically zero \citep{JouffroyR15}. Through this approach, the modulating function method allows for fixed-time estimation, while also being robust to additive disturbances \citep{KorderNR22, Byrski03, Noack_Polynomial_Estimation, Acc_DT_MFM}.

Throughout the years, several families of modulating functions have been proposed, such as the polynomial family \citep{LoebCahen65}, the spline family \citep{MaletinskySplineFrequencyDomain}, and the different hyperbolic families \citep{Acc_Hyperbolic_MF}. Moreover, it was recently shown in \cite{Acc_Algebra_MF} that modulating functions are not only well-behaved algebraically, with TMFs forming an algebra, but that it is possible to construct infinitely many new MFs from any given sufficiently smooth function, as well as by combining modulating functions using addition and multiplication.

Nonetheless, instead of relying on a predefined MF, another approach was proposed in \cite{SchmidR11} (see also \cite{Noack_History_Auxiliary_Systems_MFM}), where the modulating function is obtained as the solution to an auxiliary system. This approach became a powerful tool for PDE estimation problems \citep{Fischer_flatness_based_PDE_2020, Folke_PDE_Advection_Diffusion_Reaction} and, together with the explicit use of duality, led to the direct connection between traditional linear observers and the modulating function method, having some well-known observer structures obtained as particular cases of the MFM \citep{D_Acc_MFM_to_Obs}. Additionally, the MFM was recently extended to discrete systems in  \cite{Acc_DT_MFM}, where it was also shown that the connection between the MFM and other estimator structures also exists in discrete-time. In particular, estimation methods such as least-squares parameter estimation and Kalman filtering can, rather surprisingly, be interpreted as particular cases of the MFM, indicating that the range of applicability of this approach might be bigger than anticipated.

One question that naturally arises in this context is if the modulating function method can be used for control purposes, especially due to the duality between estimation and control tasks \citep[pp.~235-239]{Kalman_General_Theory_Control, Callier_linear_system_theory} and the growing use of auxiliary systems within the modulating function community \citep{Fischer21, D_Acc_MFM_to_Obs, Folke_PDE_Advection_Diffusion_Reaction}. So far, there have been three mentions of how the MFM can be used for input generation: a PID with a feedforward term that is obtained from parameter estimation \citep{Ndoye_Intelligent_PID_MFM}; estimation of the input with control law $\mathbf{u}=-\mathbf{K}\mathbf{x}$ using a functional observer for infinite-dimensional LTI systems \citep{Folke_MF_State_Estimation_and_Feedback}; and, recently, controller design for a first-order single-input scalar differential equation \citep{Wang_MF_Controller_Design}. Although this indicates that there is some interest in the topic, it remains drastically unexplored. In particular, \cite{Ndoye_Intelligent_PID_MFM} discusses how parameter estimates can be used in the control of a particular second-order system, without making use of auxiliary systems or duality; \cite{Folke_MF_State_Estimation_and_Feedback} does not discuss controller design, but combining the estimation step with the controller gains $\mathbf{K}$; while \cite{Wang_MF_Controller_Design} considers only first-order control-affine scalar differential equations with a single input, not making any use of duality or auxiliary systems, and without relating the proposed approach to any existing controller architectures.

To fill this knowledge gap, this paper extends the traditional MFM by defining dual modulating functions and dual modulations, naturally leading to a controller design framework for LTV MIMO systems of arbitrary order. In particular, the main contributions of this paper are:
\begin{itemize}
    \item To introduce dual modulating functions and dual modulations, as well as to discuss their properties;
    \item To formally introduce the MFM as a controller design method using the dual modulation approach;
    \item To prove that well-know controller architectures, such as full-state feedback and LTI sliding mode control, can be obtained as particular selections of dual modulations;
    \item To define and discuss the self-modulation operator and the controllability bracket, along with their relation to the existence of finite-time stabilizing controllers;
    \item To propose an MF-based control law for LTV MIMO systems that allows for fixed-time convergence, with estimates of transient trajectories and numerical errors;
    \item To illustrate the proposed concepts with different examples, including numerical simulations.
\end{itemize}

Note that, although it is possible to extend the results of this paper to discrete-time systems by considering the results in \cite{Acc_DT_MFM}, this is not done in this document for two main reasons: first, clarity of exposition but, more importantly, because it is already well-known that discrete system can be stabilized in finite-time \citep[pp.~379-381]{Ogata_Discrete_Systems}. 

After this introduction, Section \ref{sec: preliminaries} briefly introduces the modulating function method to the reader, together with the notion of dual modulating functions. Then, Section \ref{sec: from the MFM to asymptotic control} introduces dual modulations, the modulation-based control framework, and relates it to state and output feedback, as well as sliding mode control. Next, Section \ref{sec: fixed-time control and feasibility tests} discusses the existence of finite-time stabilizing controllers, introduces the self-modulation and controllability bracket tests, as well as a newly introduced fixed-time controller and its expected transient behavior. Lastly, illustrative numerical simulations are shown in Section \ref{sec: numerical simulations}, followed by some concluding remarks in Section \ref{sec: conclusion}.

%% file: 02_Preliminaries.tex
\section{Motivation and Preliminaries}\label{sec: preliminaries}

In this section, the traditional modulating function method is briefly introduced, starting with the definition of modulating functions, followed by the motivation and definition of dual modulating functions. For some illustrations of MFs, see, e.g., \cite{Acc_Hyperbolic_MF}.

\begin{defn}
    Let $\varphi: [0, T] \mapsto \mathbb{R}$ be a $\mathcal{C}^{q-1}([0, T])$ smooth function, where $q \in \mathbb{N}_{>0}$ and $T \in \mathbb{R}_{>0}$. If the function $\varphi(\tau)$ satisfies \citep{JouffroyR15}
    \begin{equation}
        \varphi^{(i)}(0)\cdot\varphi^{(i)}(T)=0
    \end{equation}
    for $i\in \{0, 1, \ldots, q-1\}$, then it is called a modulating function. In addition,
    \begin{itemize}
        \item if $\varphi^{(i)}(0)=\varphi^{(i)}(T)=0 ~\forall i$, then it is called a total modulating function (TMF);
        \item if $\varphi^{(i)}(0)=0 ~\forall i$ and $\exists i$ s.t. $\varphi^{(i)}(T) \neq 0$, then it is called a left modulating function (LMF);
        \item if $\varphi^{(i)}(T)=0 ~\forall i$ and $\exists i$ s.t. $\varphi^{(i)}(0) \neq 0$, then it is called a right modulating function (RMF).
    \end{itemize}
\end{defn}

While TMFs are often used in parameter estimation settings \citep{MaletinskySplineFrequencyDomain, Ungarala_MFM_Bioprocesses}, LMFs and RMFs are commonly used for state estimation tasks \citep{LiuKPG14, JouffroyR15, Acc_Hyperbolic_MF, Acc_DT_MFM}. In addition, LMFs and RMFs that evaluate to $1$ at the boundary are useful in later sections, and, for that reason, receive their own name.

\begin{defn}
    Let $\varphi_{\text{L}}: [0, T] \mapsto \mathbb{R}$ be an LMF and $\varphi_{\text{R}}: [0, T] \mapsto \mathbb{R}$ an RMF. If $\varphi_{\text{L}}(\tau)$ satisfies $\varphi_{\text{L}}(T)=1$, then it is called a unitary LMF. Similarly, if $\varphi_{\text{R}}(0)=1$, then $\varphi_{\text{R}}(\tau)$ is called a unitary RMF.
\end{defn}

Now, we motivate the definition of dual modulating functions with the following example.

\begin{exmp}\label{example: motivating example}
    Consider the mass-spring-damper system
    \begin{equation}\label{eq: mass-damper IO representation}
        m\ddot{\mathrm{y}} + b\dot{\mathrm{y}}+k\mathrm{y}=\mathrm{u},
    \end{equation}
    where $\mathrm{y}(t)$ is the position, $m \in \mathbb{R}_{> 0}$ is the system's inertia, $b \in \mathbb{R}_{> 0}$ is the damping ratio, $k \in \mathbb{R}_{>0}$ is the spring constant, and $\mathrm{u}(t) \in \mathbb{R}$ is the input force. Now, suppose that it is desired to go from an initial position $\mathrm{y}(t_0)$ and velocity $\dot{\mathrm{y}}(t_0)$ to a final position $\mathrm{y}_d(t_1)$ and velocity $\dot{\mathrm{y}}_d(t_1)$.
    
    By multiplying both sides of \eqref{eq: mass-damper IO representation} by a modulating function $\varphi(t_1-\tau)$ and integrating from $t_0$ to $t_1$, one obtains
    \begin{equation}
    \begin{aligned}
        \int_{t_0}^{t_1} \varphi(t_1-\tau) \mathrm{u}(\tau) \dd \tau=\int_{t_0}^{t_1} \varphi(t_1-\tau) m\ddot{\mathrm{y}}(\tau) \dd \tau \\+ \int_{t_0}^{t_1} \varphi(t_1-\tau)b\dot{\mathrm{y}}(\tau) \dd \tau + \int_{t_0}^{t_1} \varphi(t_1-\tau)k\mathrm{y}(\tau) \dd \tau.
    \end{aligned}
    \end{equation}
    Next, define $T:=t_1-t_0$, apply integration by parts, and reorganize, leading to
    \begin{equation}\label{eq: intermediate step example 1}
    \begin{aligned}
        \varphi(0) \dot{\mathrm{y}}(t_1) + \Bigl(\varphi'(0)
        + \frac{b}{m} \varphi(0) \Bigr) \mathrm{y}(t_1) = \\ \varphi(0) \dot{\mathrm{y}}_d(t_1) + \Bigl(\varphi'(0)
        + \frac{b}{m} \varphi(0) \Bigr) \mathrm{y}_d(t_1)+f(\mathrm{u}),
    \end{aligned}
    \end{equation}
    where $f(\mathrm{u}) \in \mathbb{R}$ is defined as
    \begin{equation}\label{eq: siso lti example integral equation}
        f\hspace{-0.06cm}:=\hspace{-0.13cm}\int_{t_0}^{t_1} \hspace{-0.14cm} \frac{1}{m}\varphi(t_1-\tau) \mathrm{u}(\tau) -\Bigl(\varphi''+\frac{b}{m}\varphi' +\frac{k}{m} \varphi \Bigr)(t_1-\tau) \mathrm{y}(\tau) \dd \tau \hspace{-0.05cm}.
    \end{equation} 
    Then, selecting $\mathrm{u}(t)$ such that
    \begin{equation}
    \begin{aligned}
        f(\mathrm{u})=\varphi(0) \dot{\mathrm{y}}_d(t_1) + \Bigl(\varphi'(0)
        + \frac{b}{m} \varphi(0) \Bigr) \mathrm{y}_d(t_1)\\ -\varphi(T) \dot{\mathrm{y}}(t_0)+ \Bigl( \varphi'(T) - \frac{b}{m} \varphi(T) \Bigr) \mathrm{y}(t_0)
    \end{aligned}         
    \end{equation}
    leads to \eqref{eq: intermediate step example 1} being rearranged into
    \begin{equation}
    \begin{aligned}
        \varphi(0) \dot{\mathrm{y}}(t_1) + \Bigl(\varphi'(0)
        + \frac{b}{m} \varphi(0) \Bigr) \mathrm{y}(t_1) = \\ \varphi(0) \dot{\mathrm{y}}_d(t_1) + \Bigl(\varphi'(0)
        + \frac{b}{m} \varphi(0) \Bigr) \mathrm{y}_d(t_1).
    \end{aligned}
    \end{equation}
    Since this is valid for any MF $\varphi$, it follows that $\mathrm{y}(t_1)=\mathrm{y}_d(t_1)$ and $\dot{\mathrm{y}}(t_1)=\dot{\mathrm{y}}_d(t_1)$.
\end{exmp}

This simple example illustrates the core of the proposed approach: selecting $\mathrm{u}(t)$ as a solution to an integral equation, which is dictated by a linear combination of the modulating functions and model coefficients. 

\begin{exmp}
    One direct parallel can be traced to flatness theory by noting that, if the modulating function $\varphi(\tau)$ in Example \ref{example: motivating example} is selected such that\footnote{Technically, the initial or final point of $\mathrm{y}(\tau)$ would have to be $0$ for $\varphi$ to be a modulating function in this case. Another approach would be to consider, e.g., $\Tilde{\mathrm{y}}(\tau):=\mathrm{y}(\tau)-\mathrm{y}(t_1)$ instead.}
    \begin{equation}\label{eq: output from MF flatness}
        \mathrm{y}(\tau)=\frac{1}{m}\varphi(t_1-\tau),
    \end{equation}
    Assume for simplicity that $f=0$. Then, \eqref{eq: siso lti example integral equation} is equivalent to
    \begin{equation}
    \begin{aligned}
        &\int_{t_0}^{t_1} \hspace{-0.1cm} \frac{1}{m}\varphi(t_1-\tau) \mathrm{u}(\tau) \dd \tau=\\& \int_{t_0}^{t_1} \hspace{-0.1cm} \frac{1}{m}\varphi(t_1-\tau) \Bigl(\varphi''+\frac{b}{m}\varphi' +\frac{k}{m} \varphi \Bigr)(t_1-\tau) \dd \tau,
    \end{aligned}
    \end{equation}
    which clearly has a solution given by
    \begin{equation}\label{eq: input from MF flatness}
        \mathrm{u}(\tau)=\varphi''(t_1-\tau)+\frac{b}{m}\varphi'(t_1-\tau) +\frac{k}{m} \varphi(t_1-\tau).
    \end{equation}
     Equations \eqref{eq: output from MF flatness} and \eqref{eq: input from MF flatness} are the well-known flatness approach to SISO LTI systems, having the function $\varphi(\tau)$ being the flat output \citep[p.~15]{Sira_Ramirez_Differentially_Flat_Systems}. Since this example can easily be extended to arbitrary LTI systems, flatness theory can, at least in some cases, be interpreted as a particular case of the modulating function framework.
\end{exmp}

From these examples, it is clear that the linear combination of the modulating function and its derivatives are crucial to this framework. Thus, we now formally define dual modulating functions for SISO LTI systems in input-output form, i.e.
\begin{equation}\label{eq: SISO LTI IO form}
    \mathrm{y}^{(n)}+\sum_{i=0}^{n-1} a_i \mathrm{y}^{(i)} =\sum_{i=0}^{n-1} b_j \mathrm{u}^{(j)},
\end{equation}
where $a_i, b_i \in \mathbb{R}$ for $i\in \{0, 1, \ldots, n-1\}$.

\begin{defn}
    Consider a SISO LTI system given by \eqref{eq: SISO LTI IO form} and let $\varphi:[0, T] \mapsto \mathbb{R}$ be a modulating function of order $q\geq n$. Then, a dual modulating function is defined as
    \begin{equation}\label{eq: dual modulating function definition}
        \gamma(\tau):= \varphi^{(n)}(\tau)+ \sum_{i=0}^{n-1} a_i \varphi^{(i)}(\tau).
    \end{equation}
\end{defn}

Note that the dual modulating function $\gamma(\tau)$ is the analogous of substituting $\mathrm{y}(t)$ by $\varphi(\tau)$ on the LHS of \eqref{eq: SISO LTI IO form}: a structure that naturally leads to a controllable canonical form. In addition, they are linear combinations of a modulating function and its derivatives, indicating that they are not necessarily modulating functions, as discussed in \cite{Acc_Algebra_MF}. Nonetheless, given a dual modulating function $\gamma(\tau)$, the differential equation \eqref{eq: dual modulating function definition} can always be solved for $\varphi(\tau)$.

Although the input-output form shown in \eqref{eq: SISO LTI IO form} is useful, in particular for parameter estimation problems, continuing the analysis in this framework would lead to issues when considering MIMO systems or attempting to relate the proposed approach to modern control theory without restricting the method to, e.g., systems in controllable canonical form \citep{Noack_History_Auxiliary_Systems_MFM}. For this reason, we now define the modulation operator, which is directly applicable to both scalar differential equations and state-space representations. Following \cite{D_Acc_MFM_to_Obs}, the modulation operator is defined through a convolution-based formulation.
\begin{defn} \label{def: modulation operator}
    A modulation operator is given by
    \begin{equation}\label{eq: modulation definition}
        \langle \mathbf{g}, \mathbf{x} \rangle(t_1, t_0) := \int_{t_0}^{t_1} \mathbf{g}^\top (t_1-\tau) \mathbf{x} (\tau) \dd \tau \text{,}
    \end{equation}
where $\mathbf{g}(\tau) \in \mathbb{R}^{n \times n}$ is the modulation kernel, $\mathbf{x}(\tau) \in \mathbb{R}^{n}$ is the signal being modulated, and the filtering window length is given as $T:=t_1-t_0$, with $0<t_0<t_1$.
\end{defn}

To avoid clutter, the dependency on $t_0$ and $t_1$ is omitted, and it is important to note that, to have such a modulation operator, it is sufficient that both functions within the operator are locally square-integrable. 

In addition, by applying the modulation operator to $\mathbf{x}(t)$, the time variable $t$ is changed to the dummy variable $\tau$, i.e. $\langle \mathbf{g}(\tau), \mathbf{x}(t)\rangle$ leads to \eqref{eq: modulation definition}. Hence, $\dot{\mathbf{x}}$ and $\mathbf{x}'$ are equivalent in this context.

A central property of the modulation operator is that it can transfer derivatives from one term of the operator to the other.

\begin{lem}\label{lemma: transfer of derivatives}
    Let $\mathbf{g}(\tau) \in \mathbb{R}^{n \times n}$ be differentiable in $[0, T]$ and $\mathbf{x}(\tau) \in \mathbb{R}^n$ be differentiable in $[t_0, t_1]$. Then, applying modulation operator \eqref{eq: modulation definition} to the derivative of a signal $\mathbf{x}(\tau)$ transfers the derivative to the modulation kernel, i.e.
    \begin{equation}
        \langle \mathbf{g}, \mathbf{x}' \rangle=\mathbf{g}^\top(0) \mathbf{x} (t_1) -\mathbf{g}^\top(T) \mathbf{x} (t_0) + \langle \mathbf{g}', \mathbf{x}\rangle \text{.}
    \label{eq: transfer of derivatives}
    \end{equation}
\end{lem}
    
\begin{pf}
    Assuming that $\mathbf{g}(\tau) \in \mathbb{R}^{n \times n}$ is differentiable in $[0, T]$ and $\mathbf{x}(\tau) \in \mathbb{R}^n$ is differentiable in $[t_0, t_1]$, start from Definition \ref{def: modulation operator}, consider the modulation of $\mathbf{x}'$, i.e. $\langle \mathbf{g}, \mathbf{x}' \rangle$, and note that
    \begin{equation*}
        \frac{\dd ~\mathbf{g}(t_1-\tau)}{\dd \tau}=-\mathbf{g}'(t_1-\tau).
    \end{equation*}
    Then, by applying integration by parts, expression \eqref{eq: transfer of derivatives} is obtained. \hfill $\blacksquare$
\end{pf}

The class of systems discussed in the remaining of this paper are linear time-varying systems, i.e.
\begin{equation}
    \dot{\mathbf{x}}(t)=\mathbf{A}(t) \mathbf{x}(t) +\mathbf{B}(t) \mathbf{u}(t)\label{eq: state eq}
\end{equation}
for  $\mathbf{A}(t) \in \mathbb{R}^{n \times n}$ and $\mathbf{B}(t) \in \mathbb{R}^{n \times m}$, which are assumed to be square integrable and continuous with respect to $t$, having the states assumed to be known at each time instant. For simplicity and clarity of exposition, the time dependency is often omitted, but it should be clear from the context what is and is not time-varying.

As a brief reminder to the reader, a property that is largely exploited throughout this paper is the adjoint modulation of time-varying real matrices, i.e.
\begin{equation}
    \langle \mathbf{g}(\tau), \mathbf{A}(\tau) \mathbf{x}(\tau)\rangle=\langle \mathbf{A}^\top(t_1-\tau) \mathbf{g}(\tau), \mathbf{x}(\tau) \rangle=:\langle \mathbf{A}^* \mathbf{g}, \mathbf{x} \rangle
    \label{eq: adjoint operator}
\end{equation}
which naturally follows from Definition \ref{def: modulation operator}. Lastly, one directly deduces that $\mathbf{A}^*=\mathbf{A}^\top$ for time-invariant $\mathbf{A}$, as used in most of \cite{D_Acc_MFM_to_Obs}. 

%% file: 03_Asymptotic_Control.tex
\section{Dual Modulations and Asymptotic Control}\label{sec: from the MFM to asymptotic control}

Due to the inherent duality between estimation and control, we now define the dual modulation kernel, which is the main object of this section.

\begin{defn}
    Consider modulation operator \eqref{eq: modulation definition} and a system given by \eqref{eq: state eq}. Then, a dual modulation kernel $\bm{\uplambda}: [0, T] \mapsto \mathbb{R}^{n \times n}$ is defined as
    \begin{equation}\label{eq: dual modulation definition}
        \bm{\uplambda}(\tau):=\mathbf{g}'(\tau)-\mathbf{A}^*(\tau) \mathbf{g}(\tau).
    \end{equation}
\end{defn}

Note the resemblance between \eqref{eq: dual modulating function definition} and \eqref{eq: dual modulation definition}: similar to dual modulating functions, the dual modulation kernel is obtained from a linear combination of the modulation kernel and its derivative. Consequently, there are no restrictions on $\bm{\uplambda}$ besides integrability. More importantly, dual modulations lead to an auxiliary system that is dual to the state equation \eqref{eq: state eq}, as seen below.

\begin{prop}\label{prop: input modulation equal state modulation}
    Given modulation operator \eqref{eq: modulation definition}, a system described by \eqref{eq: state eq}, and a dual modulation \eqref{eq: dual modulation definition}, the relation
    \begin{equation}
        \mathbf{g}^{\top}(0)\mathbf{x}(t_1)= \mathbf{g}^{\top}(T) \mathbf{x}(t_0)+\langle\mathbf{v}, \mathbf{u}\rangle - \langle\bm{\uplambda} , \mathbf{x}\rangle
        \label{eq: input modulation as state modulation}
    \end{equation}
    is true for $\mathbf{g}(\tau) \in \mathbb{R}^{n \times n}$ and $\mathbf{v}(\tau) \in \mathbb{R}^{m \times n}$ satisfying the dual state equation
    \begin{gather}
        \mathbf{g}'=\mathbf{A}^* \mathbf{g} + \bm{\uplambda} \label{eq: aux system state eq}\\
        \mathbf{v}=\mathbf{B}^* \mathbf{g} \label{eq: aux system out eq} \text{,}
    \end{gather}
    having the dual modulation $\bm{\uplambda}(\tau) \in \mathbb{R}^{n \times n}$ as the design functional.
\end{prop}

\begin{pf}
    By applying modulation operator \eqref{eq: modulation definition} to \eqref{eq: state eq}
    \begin{equation}
        \langle\mathbf{g}, \dot{\mathbf{x}}\rangle=\langle\mathbf{g}, \mathbf{A} \mathbf{x}\rangle + \langle\mathbf{g},\mathbf{B} \mathbf{u}\rangle\text{,}
    \end{equation}
    and using Lemma \ref{lemma: transfer of derivatives} on the LHS, one obtains
    \begin{equation}
         \mathbf{g}^\top(0) \mathbf{x}(t_1) -\mathbf{g}^\top(T) \mathbf{x}(t_0) + \langle\mathbf{g}',\mathbf{x}\rangle=\langle\mathbf{g},\mathbf{A} \mathbf{x}\rangle + \langle\mathbf{g},\mathbf{B} \mathbf{u}\rangle .
    \end{equation}
    It is then possible to group the modulations as
    \begin{equation*}
         \mathbf{g}^\top(0) \mathbf{x}(t_1) -\mathbf{g}^\top(T) \mathbf{x}(t_0) + \langle\mathbf{g}'-\mathbf{A} ^*\mathbf{g} ,\mathbf{x}\rangle = \langle\mathbf{B}^*\mathbf{g}, \mathbf{u}\rangle
    \end{equation*}
    and, by using a dual modulation kernel \eqref{eq: dual modulation definition}, one obtains \eqref{eq: aux system state eq}-\eqref{eq: aux system out eq} and \eqref{eq: input modulation as state modulation}. \hfill $\blacksquare$
\end{pf}

Thus far, the use of auxiliary systems within the modulating function method has focused on, for example, ensuring linear independence of modulations in a parameter estimation setting \citep{SchmidR11} and providing a bridge between observers and the MFM \citep{D_Acc_MFM_to_Obs}. However, Proposition \ref{prop: input modulation equal state modulation} shows that this duality-based approach is also useful for control-related problems.

\begin{exmp}\label{example: state-transition matrix}
    As a simple example, consider state equation \eqref{eq: state eq} and apply Proposition \ref{prop: input modulation equal state modulation}, leading to
    \begin{equation}\label{eq: state transition matrix intermediate step}
        \mathbf{g}^\top(0)\mathbf{x}(t_1)=\mathbf{g}^\top(T) \mathbf{x}(t_0)+\langle \mathbf{v}, \mathbf{u} \rangle - \langle\bm{\uplambda}, \mathbf{x}\rangle\text{.} 
    \end{equation}
    Now, by selecting the dual modulation kernel as
    \begin{equation}
        \bm{\uplambda}(\tau)=\bm{0}
    \end{equation}
    and an initial condition $\mathbf{g}^\top(0)=\mathbf{g}^*(t_1)$ such that $\rank~\mathbf{g}(0)=n$, it is possible to rearrange \eqref{eq: state transition matrix intermediate step} and expand the modulation operator to obtain
    \begin{equation}
        \begin{aligned}
            \mathbf{x}(t_1)=&\mathbf{g}^{-*}(t_1)\mathbf{g}^*(t_0) \mathbf{x}(t_0)\\&+\int_{t_0}^{t_1} \mathbf{g}^{-*}(t_1) \mathbf{g}^*(\tau) \mathbf{B}(\tau) \mathbf{u}(\tau)\dd \tau \text{,} 
        \end{aligned}
    \end{equation}
    having the notation $\mathbf{g}^{-*}(\tau):=\left(\mathbf{g}^*(\tau)\right)^{-1}$ introduced.
    
    Lastly, by defining $\bm{\Phi}(t_1, t_0):=\mathbf{g}^{-*}(t_1)\mathbf{g}^*(t_0)$, the well known LTV state-space solution is obtained as
    \begin{equation}
        \mathbf{x}(t_1)=\bm{\Phi}(t_1, t_0) \mathbf{x}(t_0) +\int_{t_0}^{t_1} \bm{\Phi}(t_1, \tau) \mathbf{B}(\tau) \mathbf{u}(\tau)\dd \tau \text{,} 
    \end{equation}
    which clearly satisfies $\bm{\Phi}(t_2, t_1)\bm{\Phi}(t_1, t_0)=\bm{\Phi}(t_2, t_0)$ and $\frac{\dd }{\dd t} \bm{\Phi}(t, t_0)=\mathbf{A} \bm{\Phi}(t, t_0)$. \hfill $\diamond$
\end{exmp}

Example \ref{example: state-transition matrix} shows that the modulation operator can be used not only to deal with estimation problems, but also to, e.g., solve state equations via dual modulations. More importantly, this indicates that the MFM can be used to provide a somewhat different perspective to control theory: control through modulations.

Despite being a simple construction, the dual modulation approach from Proposition \ref{prop: input modulation equal state modulation} allows for direct analysis of the set of reachable and controllable states of continuous-time systems. In particular, reachable and controllable states can be directly discussed in terms of the modulation kernels.

For the sake of completeness, we now recall the definition of reachability and controllability.

\begin{defn}\label{def: reachable and controllable states}
    A state $\mathbf{x}(t_1)$ is said to be reachable on $[t_0, t_1]$, for $t_0\leq t \leq t_1$, if there exists an input $\mathbf{u}(t)$ that drives the states from $\mathbf{x}(t_0)=\bm{0}$ to $\mathbf{x}(t_1)$. Similarly, a state $\mathbf{x}(t_0)$ is said to be controllable to zero on $[t_0, t_1]$, for $t_0\leq t \leq t_1$, if there exists an input $\mathbf{u}(t)$ that drives the states from $\mathbf{x}(t_0)$ to $\mathbf{x}(t_1)=\bm{0}$. \citep[p.~227]{Callier_linear_system_theory}.
\end{defn}

By simply inspecting \eqref{eq: input modulation as state modulation} and comparing it to Definition \ref{def: reachable and controllable states}, two results are trivially obtained.

\begin{cor} \label{corollary: reachability from modulation}
    Given modulation operator \eqref{eq: modulation definition}, a system described by \eqref{eq: state eq}, a dual modulation \eqref{eq: dual modulation definition}, and auxiliary system \eqref{eq: aux system state eq}-\eqref{eq: aux system out eq}, the set of reachable states is given as
    \begin{equation}\label{eq: reachability from modulation}
        \mathfrak{R}=\left\{\mathbf{x}(t_1) \in \mathbb{R}^{n}~|~ \mathbf{x}(t_1)=\mathbf{g}^{-*}(t_1) \left( \langle\mathbf{v},\mathbf{u}\rangle-\langle\bm{\uplambda}, \mathbf{x}\rangle \right) \right\}\text{,}
    \end{equation}
    for $\rank~ \mathbf{g}^*(t_1)=n$.
\end{cor}

\begin{cor} \label{corollary: controllability to zero from modulation}
    Given modulation operator \eqref{eq: modulation definition}, a system described by \eqref{eq: state eq},a dual modulation \eqref{eq: dual modulation definition}, and auxiliary system \eqref{eq: aux system state eq}-\eqref{eq: aux system out eq}, the set of controllable-to-zero states is given as
    \begin{equation}\label{eq: controllability to zero from modulation}
        \mathfrak{C}=\{\mathbf{x}(t_0) \in \mathbb{R}^{n}~|~ \mathbf{x}(t_0)=-\mathbf{g}^{-*}(t_0)\left( \langle\mathbf{v},\mathbf{u} \rangle - \langle\bm{\uplambda},\mathbf{x}\rangle\right) \}\text{,}
    \end{equation}
    for $\rank ~\mathbf{g}^*(t_0)=n$.
\end{cor}

Since one can solve the auxiliary system in either time direction for a given initial or final condition, the main requirement to have all states reachable or controllable is the same: the existence of $\mathbf{u}(t)$ such that $\langle\mathbf{v}, \mathbf{u}\rangle-\langle \bm{\uplambda}, \mathbf{x} \rangle$ spans $\mathbb{R}^n$. By then considering control laws resulting from\footnote{This is the analogous of considering $f(\mathrm{u})=0$ in \eqref{eq: siso lti example integral equation}.} $\langle\mathbf{v}, \mathbf{u}\rangle-\langle \bm{\uplambda}, \mathbf{x} \rangle=\bm{0}$, the existence of solutions is a problem that naturally emerges.

\subsection{Existence of Solutions}

Using concepts from linear operator theory and functional analysis, a necessary and sufficient condition for the existence of such control laws is obtained.

\begin{prop}\label{prop: trajectory constraint}
    Consider modulation operator \eqref{eq: modulation definition}, a system described by \eqref{eq: state eq}, a dual modulation \eqref{eq: dual modulation definition}, auxiliary system \eqref{eq: aux system state eq}-\eqref{eq: aux system out eq}, and let $\mathsf{H}(\mathbf{x})$ and $\mathsf{R}(\mathbf{u})$ be the linear operators
    \begin{equation}\label{eq: H and R map definition}
        \mathsf{R}(\mathbf{u}):=\langle\mathbf{v}, \mathbf{u}\rangle\text{,} \; \; \; \mathsf{H}(\mathbf{x}):=\langle\bm{\uplambda}, \mathbf{x}\rangle\text{.}
    \end{equation}
    Then, an input
    \begin{equation}\label{eq: input from modulation for natural trajectories}
        \mathbf{u}(t)=\{\mathbf{u}(t) \in \mathbb{R}^{m} | \langle\mathbf{v},  \mathbf{u} \rangle=\langle\bm{\uplambda}, \mathbf{x}\rangle \}
    \end{equation}
    exists if and only if
    \begin{equation}
        \text{range}(\mathsf{H}) \subseteq \text{range}(\mathsf{R}),  
    \end{equation}
    and, by selecting $\mathbf{g}^\top(0)=\mathbf{g}^*(t_1)$ such that $\rank~\mathbf{g}^*(t_1)=n$,  $\mathbf{x}(t_1)$ is obtained as
    \begin{equation}\label{eq: trajectory from modulation}
        \mathbf{x}(t_1)=\bm{\Psi}(t_1, t_0; \bm{\uplambda}) \mathbf{x}(t_0),
    \end{equation}
    for $\bm{\Psi}(t_1, t_0; \bm{\uplambda}):= \mathbf{g}^{-*}(t_1)\mathbf{g}^*(t_0)$.
\end{prop}

\begin{pf}
    Starting from \eqref{eq: input modulation as state modulation}, i.e.
    \begin{equation}\label{eq: input selection for natural trajectories intermediate step}
        \mathbf{g}^*(t_1)\mathbf{x}(t_1)=\mathbf{g}^*(t_0) \mathbf{x}(t_0)+\langle\mathbf{v}, \mathbf{u}\rangle-\langle\bm{\uplambda}, \mathbf{x}\rangle\text{,}
    \end{equation}
    consider the situation where it is desired to select the input $\mathbf{u}(t)$ given a dual modulation $\bm{\uplambda}(\tau)$ such that
    \begin{equation}
        \langle\mathbf{v}, \mathbf{u}\rangle-\langle\bm{\uplambda}, \mathbf{x}\rangle =\bm{0}
    \end{equation}
    or, equivalently,
    \begin{equation}\label{eq: input selection from modulation intermediate step}
        \langle\mathbf{v}, \mathbf{u} \rangle =\langle\bm{\uplambda}, \mathbf{x}\rangle\text{.}
    \end{equation}
    Let $\mathsf{H}(\mathbf{x})$ and $\mathsf{R}(\mathbf{u})$ be the linear operators defined as \eqref{eq: H and R map definition}. Then, the integral equation can be rewritten as $\mathsf{H}(\mathbf{x})=\mathsf{R}(\mathbf{u})$, and it directly follows from the definition of the range of linear operators that $\mathsf{H}(\mathbf{x})=\mathsf{R}(\mathbf{u})$ is solvable for $\mathbf{u}(t)$ given any $\mathbf{x}(t)$ iff $\text{range}(\mathsf{H}) \subseteq \text{range}(\mathsf{R})$.
    
    Lastly, by substituting \eqref{eq: input selection from modulation intermediate step} in \eqref{eq: input selection for natural trajectories intermediate step}, one obtains
    \begin{equation}
        \mathbf{g}^*(t_1) \mathbf{x}(t_1)=\mathbf{g}^*(t_0) \mathbf{x}(t_0)\text{,}
    \end{equation}
    which, by selecting any $\mathbf{g}^*(t_1)$ such that $\rank~ \mathbf{g}^*(t_1)=n$ and defining $\bm{\Psi}(t_1, t_0; \bm{\uplambda}):= \mathbf{g}^{-*}(t_1)\mathbf{g}^*(t_0)$, has \eqref{eq: trajectory from modulation} as the unique solution. \hfill $\blacksquare$
\end{pf}

An important result is shown in Proposition \ref{prop: trajectory constraint}: the final state can be selected through the matrix $\bm{\Psi}(t_1, t_0; \bm{\uplambda}):=\mathbf{g}^{-*}(t_1) \mathbf{g}^*(t_0)$ using the dual modulation $\bm{\uplambda}$, as long as the integral equation $\langle\mathbf{v},  \mathbf{u}\rangle=\langle\bm{\uplambda}, \mathbf{x}\rangle$ is solvable for $\mathbf{u}$. For this reason, the matrix $\bm{\Psi}(t_1, t_0; \bm{\uplambda})$, having its dependency on $\bm{\uplambda}$ often omitted, is similar in spirit to the state transition matrix: it takes the state $\mathbf{x}(t_0)$ to $\mathbf{x}(t_1)$.

\subsection{From the MFM to Asymptotic Control}

One approach to ensure that $\text{range}(\mathsf{H}) \subseteq \text{range}(\mathsf{R})$ is to select $\mathsf{H}$ as a composition between an arbitrary linear map $\mathsf{M}$ and $\mathsf{R}$, i.e. $\mathsf{H}=\mathsf{R} \mathsf{M}$. Surprisingly, the structure for linear full-state feedback is directly obtained through this approach.

\begin{thm}\label{theorem: asymptotic full state feedback}
    Given modulation operator \eqref{eq: modulation definition} and a system described by \eqref{eq: state eq}, the full-state feedback control law $\mathbf{u}=-\mathbf{K} \mathbf{x}$ is obtained from \eqref{eq: input from modulation for natural trajectories} by selecting the dual modulation $\bm{\uplambda}=-\mathbf{K}^* \mathbf{v}$, for $t_1=t$, $t_0=0$, and $\rank~\mathbf{g}(0)=n$, having $\mathbf{K}(t) \in \mathbb{R}^{m \times n}$ as a tuning parameter.
\end{thm}

\begin{pf}
    Starting from Proposition \ref{prop: input modulation equal state modulation}, select $t_1=t$, $t_0=0$, and $\bm{\uplambda}=-\mathbf{K}^* \mathbf{v}$ to obtain
    \begin{equation}
        \mathbf{g}'=\mathbf{A}^* \mathbf{g} -\mathbf{K}^* \mathbf{v} \text{,}
    \end{equation}
    which can be rewritten as
    \begin{equation}
        \mathbf{g}'=\left(\mathbf{A}- \mathbf{B}\mathbf{K}\right)^* \mathbf{g} \text{,}
    \end{equation}
    requiring $\left( \mathbf{A} -\mathbf{B} \mathbf{K} \right)$ stable for $\mathbf{g}$ to be square integrable for all $t\in \mathbb{R}\cup \infty$, a condition that can be guaranteed if the system is fully controllable.
    
    By then using Proposition \ref{prop: trajectory constraint}, the input must be obtained as the solution to
    \begin{equation}
        \langle\mathbf{v}, \mathbf{u}\rangle=\langle\bm{\uplambda}, \mathbf{x}\rangle\text{,}
    \end{equation}
    which, by using the selected $\bm{\uplambda}=-\mathbf{K}^* \mathbf{v}$, leads to
    \begin{equation}
        \langle\mathbf{v}, \mathbf{u}\rangle=\langle-\mathbf{K}^* \mathbf{v}, \mathbf{x}\rangle
    \end{equation}
    or, equivalently,
    \begin{equation}
        \mathsf{R}(\mathbf{u}):=\langle\mathbf{v}, \mathbf{u}\rangle=\langle \mathbf{v}, -\mathbf{K}\mathbf{x}\rangle=:\mathsf{R}(-\mathbf{K}\mathbf{x})\text{,}
    \end{equation}
    which clearly satisfies $\text{range}(\mathsf{H}) \subseteq \text{range}(\mathsf{R})$. Moreover, it follows that
    \begin{equation}
        \langle \mathbf{v}, \mathbf{u}+\mathbf{K}\mathbf{x} \rangle = \bm{0},
    \end{equation}    
    and there are two ways for the equation above to be identically $\bm{0}$: if $(\mathbf{u}+\mathbf{K} \mathbf{x}) \in \text{null}(\mathsf{R})$, or $(\mathbf{u}+\mathbf{K} \mathbf{x})=\bm{0}$. Using the latter approach, the input must be selected as $\mathbf{u}=-\mathbf{K} \mathbf{x}$ for some $\mathbf{K}(t) \in \mathbb{R}^{m \times n}$. \hfill $\blacksquare$

\end{pf}

\begin{exmp}
    Consider the output feedback scenario for $\mathbf{y}=\mathbf{C} \mathbf{x}$, where $\mathbf{C}(t) \in \mathbb{R}^{p \times n}$. Then, it suffices to consider the dual modulation $\bm{\uplambda}=-\mathbf{C}^*\mathbf{K}^* \mathbf{v}$, where $\mathbf{K}(t) \in \mathbb{R}^{n \times p}$, and apply Proposition \ref{prop: trajectory constraint} in the same way as in the proof of Theorem \ref{theorem: asymptotic full state feedback} to obtain
    \begin{equation}
        \langle\mathbf{v}, \mathbf{u} \rangle=\langle-\mathbf{C}^*\mathbf{K}^*\mathbf{v}, \mathbf{x} \rangle
    \end{equation}
    or, equivalently,
    \begin{equation}
        \langle\mathbf{v}, \mathbf{u}+\mathbf{K}\mathbf{y}\rangle=\bm{0},
    \end{equation}
    leading to the output feedback control law $\mathbf{u}=-\mathbf{K}\mathbf{y}$. Thus, any controller obtained from the dual modulation $\bm{\uplambda}=-\mathbf{K}^*\mathbf{v}$ or $\bm{\uplambda}=-\mathbf{C}^*\mathbf{K}^*\mathbf{v}$ satisfy $\text{range}(\mathsf{H})\subseteq \text{range}(\mathsf{R})$, and the control laws $\mathbf{u}=-\mathbf{K}\mathbf{x}$ and $\mathbf{u}=-\mathbf{K}\mathbf{y}$ are inputs given by \eqref{eq: input from modulation for natural trajectories}. \hfill $\diamond$
\end{exmp}

Another important class of controllers that can be obtained from the proposed modulation-based framework is the family of sliding mode controllers for linear systems or systems linearized by feedback, i.e. \citep[p.~60, pp.~94-95]{Utkin_Sliding_Mode}
\begin{equation}\label{eq: sliding mode controller definition}
    \mathbf{u}=-\mathbf{M}(\mathbf{x})\text{sign}(\mathbf{s}),
\end{equation}
for $\mathbf{s} \in \mathbb{R}^{m}$ being a linear combination of the states, i.e. \citep[pp.~93-94]{Utkin_Sliding_Mode}
\begin{equation}\label{eq: sliding surface defintion}
    \mathbf{s}=\mathbf{T} \mathbf{x},
\end{equation}
where $\mathbf{T} \in \mathbb{R}^{m \times n}$ and $\mathbf{M}(\mathbf{x}) \in \mathbb{R}^{m \times m}$ are tuning parameters.

\begin{thm}
    Given modulation operator \eqref{eq: modulation definition} and a system described by \eqref{eq: state eq}, the sliding mode control law \eqref{eq: sliding mode controller definition}-\eqref{eq: sliding surface defintion} is obtained from \eqref{eq: input from modulation for natural trajectories} by selecting the dual modulation $\bm{\uplambda}=-\mathbf{K}^* \mathbf{v}$, for $t_1=t$, $t_0=0$, and $\rank~\mathbf{g}(0)=n$, having $\mathbf{K}(t) \in \mathbb{R}^{m \times n}$ selected as $\mathbf{K}=\mathbf{M}(\mathbf{x}) \bm{S}(\mathbf{s}) \mathbf{T}$, where $\bm{S}(\mathbf{s}) \in \mathbb{R}^{m \times m}$ is a diagonal matrix defined as
    \begin{equation}\label{eq: definition S matrix}
        S_{i,i}(\mathbf{s}):=\left\{ \begin{array}{lcr}
             |\mathrm{s}_i|^{-1}  &, & \mathrm{s}_i \neq 0 \\
             0  &, & \mathrm{s}_i = 0
        \end{array} \right. ,
    \end{equation}
    having $\mathbf{T} \in \mathbb{R}^{m \times n}$ and $\mathbf{M}(\mathbf{x}) \in \mathbb{R}^{m \times m}$ as tuning parameters.
\end{thm}
\begin{pf}
    Starting from Proposition \ref{prop: input modulation equal state modulation}, select $t_1=t$, $t_0=0$, and $\bm{\uplambda}=-\mathbf{K}^* \mathbf{v}$, following the same steps as in the proof of Theorem \ref{theorem: asymptotic full state feedback}, to obtain
    \begin{equation}\label{eq: sliding mode input first step}
        \mathbf{u}=-\mathbf{K}\mathbf{x}.
    \end{equation}
    for $\mathbf{K} \in \mathbb{R}^{m \times n}$.
    
    Then, select $\mathbf{K}$ as $\mathbf{K}=\mathbf{M}(\mathbf{x}) \bm{S}(\mathbf{s}) \mathbf{T}$, for $\bm{S}(\mathbf{s}) \in \mathbb{R}^{m \times m}$ defined as \eqref{eq: definition S matrix}, and having $\mathbf{T} \in \mathbb{R}^{m \times n}$ and $\mathbf{M}(\mathbf{x}) \in \mathbb{R}^{m \times m}$ as tuning parameters. Next, define $\mathbf{s} \in \mathbb{R}^{m}$ as in \eqref{eq: sliding surface defintion} and note that \eqref{eq: sliding mode input first step} can be rewritten as
    \begin{equation}
        \mathbf{u}=-\mathbf{M}(\mathbf{x}) \bm{S}(\mathbf{s}) \mathbf{s}.
    \end{equation}
    Finally, note that the product $\bm{S}(\mathbf{s})\mathbf{s}$ is, in fact, $\text{sign}(\mathbf{s})$ to obtain \eqref{eq: sliding mode controller definition}. \hfill $\blacksquare$
\end{pf}

\begin{rem}
    Sliding mode observers for linear systems \citep[pp.~106-107]{Yuri_Sliding_mode_control_and_observation} can also be obtained as particular cases of the MFM in a similar manner. In particular, they are obtained from equation (71) of \cite{D_Acc_MFM_to_Obs} for the selection $\mathbf{F}=\mathbf{L} \bm{S}$, where $\mathbf{L}$ is related to the observer gains and $\bm{S}$ is defined analogous to \eqref{eq: definition S matrix}.
\end{rem}

Thus, at least some traditional asymptotic controllers can be interpreted as particular selections of dual modulations with $t_1=t$, $t_0=0$, and $T=\infty$, having the stability associated with the selection of the gain matrix $\mathbf{K}$. This provides a significant link between the modulating function method and traditional control theory: while observer structures can be obtained through state feedback in the dual auxiliary system, controller structures can be obtained by selecting dual modulations via output feedback.

%% file: 04_Fixed_Time_Control.tex
\section{Fixed-Time Control and Feasibility Tests}\label{sec: fixed-time control and feasibility tests}

Besides asymptotic controllers, Proposition \ref{prop: trajectory constraint} also allows for the development of non-asymptotic controllers, which must then satisfy $\bm{\Psi}(t_1, t_0)=\bm{0}$. In particular, this leads to a necessary and sufficient condition for the finite-time stabilization of LTV systems.

\begin{rem}
    To avoid circularity in the argument, the state-transition matrix in the remaining of this document should be interpreted using the traditional approach, i.e. using fundamental matrices \citep[pp.~107-110]{ChenLinearSystems}. In addition, the initial condition is selected as $\mathbf{g}(0)=\mathbf{I}$ for the sake of simplicity but without any loss of generality.
\end{rem}

\begin{prop}\label{prop: Gamma(T)=-I}
    Consider a system given by \eqref{eq: state eq}, let $\bm{\uplambda}(\tau) \in \mathbb{R}^{n \times n}$ be a dual modulation kernel associated with $\mathbf{g}(\tau) \in \mathbb{R}^{n \times n}$ and initial condition $\mathbf{g}(0)=\mathbf{I}$, let $\bm{\Phi}(t_1, t_0) \in \mathbb{R}^{n \times n}$ be the state-transition matrix, and consider $\bm{ \Lambda}(\tau) \in \mathbb{R}^{n \times n}$ defined as
    \begin{equation}\label{eq: Gamma of h definition}
        \bm{ \Lambda}(\tau):=\int_0 ^ \tau \bm{\Phi}^\top(t_1, t_1- \sigma) \bm{\uplambda}( \sigma) \dd  \sigma.
    \end{equation}
    Then, a control law obtained from \eqref{eq: input from modulation for natural trajectories} drives the system to the origin in a finite-time $T \in \mathbb{R}_{>0}$ if and only if
    \begin{equation}\label{eq: Gamma(T)=-I}
        \bm{ \Lambda}(T)=-\mathbf{I}.
    \end{equation}
\end{prop}
\begin{pf}
    Start from Proposition \ref{prop: input modulation equal state modulation} and consider the solution of auxiliary system \eqref{eq: aux system state eq} with initial condition $\mathbf{g}(0)=\mathbf{I}$, i.e.
    \begin{equation}\label{eq: solution of LTV auxiliary system with h intermediate step}
        \mathbf{g}(\tau)=\bm{\Phi}^\top(t_1-\tau, t_1) + \int_0^\tau \bm{\Phi}^\top (t_1-\tau, t_1- \sigma) \bm{\uplambda}( \sigma) \dd  \sigma,
    \end{equation}
    where it was used that the state transition matrix associated with $\mathbf{A}^*(\tau)$ is given as $\bm{\Phi}^* (\tau,  \sigma)=\bm{\Phi}^\top (t_1-\tau, t_1- \sigma)$, for $\bm{\Phi} (\tau,  \sigma)$ being the state transition matrix associated with $\mathbf{A}$ \citep[p.~236]{Callier_linear_system_theory}. \\
    Then, define $\bm{ \Lambda}(\tau)$ as \eqref{eq: Gamma of h definition} and note that \eqref{eq: solution of LTV auxiliary system with h intermediate step} can be rewritten as
    \begin{equation} \label{eq: solution of LTV auxiliary system with h using Gamma}
        \mathbf{g}(\tau)=\bm{\Phi}^\top(t_1-\tau, t_1) + \bm{\Phi}^\top(t_1-\tau, t_1) \bm{ \Lambda}(\tau).
    \end{equation}
    Next, to satisfy $\bm{\Psi}(t_1, t_0)=(\mathbf{g}^{\top}(0))^{-1} \mathbf{g}^\top(T)= \bm{0}$, it must follow that $\mathbf{g}(T)=\bm{0}$ and, consequently,
    \begin{equation}
        \bm{\Phi}^\top(t_1-\tau, t_1) \bm{ \Lambda}(T)=-\bm{\Phi}^\top(t_1-\tau, t_1).
    \end{equation}
    It is clearly sufficient that $\bm{ \Lambda}(T)=-\mathbf{I}$. For the necessity, simply note that the state transition matrix is always invertible for any finite time and thus it cannot be true that $\bm{\Phi}^\top(t_0, t_1)=\bm{0}$ for any $T=t_1-t_0<\infty$. \hfill $\blacksquare$
\end{pf}

Proposition \ref{prop: Gamma(T)=-I} provides an interesting connection: while the convergence on estimation problems is related to the value of the modulating function and its derivatives, convergence in the control problem is related to the integral of the dual modulation. Moreover, a crucial consequence of this result is that, if $\bm{ \Lambda}(T)=-\mathbf{I}$, then a fixed-time stabilizing controller may be constructed using a moving-window scheme, as later discussed.

\begin{exmp}\label{example: dirac delta without controllability bracket}
    In both theoretical and practical settings, it is crucial to be able to easily evaluate, given a particular selection of $\bm{\uplambda}$, if \eqref{eq: Gamma(T)=-I} and the range condition provided in Proposition \ref{prop: trajectory constraint} are satisfied. As a simple example, consider the dual modulation $\bm{\uplambda}(\tau)=-\bm{\Phi}^\top(t_1-\tau, t_1)\delta(\tau-t_1)$, for $\delta(\tau)$ being the Dirac delta. Then, it trivially follows that
    \begin{equation}
        \bm{ \Lambda}(\tau)=\left\{ \begin{array}{ccc}
             \bm{0} &, & \tau \in [0, T) \\
             -\bm{I} &, & \tau=T
        \end{array} \right.
    \end{equation}
    and, consequently, \eqref{eq: Gamma(T)=-I} is satisfied. To check for the existence of a control law given by \eqref{eq: input from modulation for natural trajectories}, note that
    \begin{equation}
        \mathsf{H}(\mathbf{x})=-\langle \bm{\Phi}^\top(t_1-\tau, t_1)\delta(\tau-t_1), \mathbf{x} \rangle
    \end{equation}
    or, equivalently, \citep[p.~236]{Callier_linear_system_theory}
    \begin{equation}
        \mathsf{H}(\mathbf{x})=-\langle \delta(\tau-t_1), \bm{\Phi}(\tau, t_0) \mathbf{x} \rangle.
    \end{equation}
    Then, expand the definition of the modulation operator
    \begin{equation}
        \mathsf{H}(\mathbf{x})=-\int_{t_0}^{t_1} \delta(-\tau) \bm{\Phi}(\tau, t_0) \mathbf{x}(\tau) \dd \tau
    \end{equation}
    and consider two separate cases: if $0 \in [t_0, t_1]$ and if $0 \notin [t_0, t_1]$. If $0 \notin [t_0, t_1]$, then $\mathsf{H}(\mathbf{x})=\bm{0}$, which is trivially within the range of $\mathsf{R}(\mathbf{u})$. In the second case, i.e. $0 \in [t_0, t_1]$, it follows that $\mathsf{H}(\mathbf{x})=\bm{\Phi}(0, t_0) \mathbf{x}(0)$, and it is necessary to compute $\mathbf{g}(\tau)$, $\mathbf{v}(\tau)$, and the range of $\mathsf{R}(\mathbf{u})$ to obtain a conclusion. \hfill $\diamond$
\end{exmp}

Note that evaluating the existence of inputs given by \eqref{eq: input from modulation for natural trajectories} by explicitly computing the range is not necessarily a simple task, even for simple dual modulations. Thus, it is of theoretical and practical importance to have simpler methods to evaluate if, given a selection of $\bm{\uplambda}$, there exists a control law that satisfies \eqref{eq: input from modulation for natural trajectories}. For this reason, two different tests are now proposed.

\subsection{Self-Modulation and Controllability Bracket Tests}

A common way to check for controllability in LTV systems is the reachability gramian. In the MF-based control context, one operator that will prove useful and is later shown to be deeply related to the reachability gramian is the self-modulation operator, which is now introduced.

\begin{defn}\label{def: self modulation}
    Given modulation operator \eqref{eq: modulation definition}, a self-modulation operator is defined as
    \begin{equation}\label{eq: self modulation def}
        \mathsf{S}(\mathbf{g}):=\int_{0}^{T} \mathbf{g}^\top(\tau) \mathbf{g}(\tau) \dd \tau = \langle\mathbf{g}, \mathbf{g}(t_1-\tau)\rangle\text{.}
    \end{equation}
\end{defn}

Note that the self-modulation operator $\langle\mathbf{g}, \mathbf{g}(t_1-\tau)\rangle$ leads to, in fact, a gram matrix. Consequently, the self-modulation operator takes in a modulation kernel and returns a linear operator, from which the range of the map $\mathsf{R}(\mathbf{u})$ can be checked. This provides a simple way to evaluate existence of solutions, as formalized in the following proposition.

\begin{prop}\label{proposition: reachability gramian}
    Consider modulation operator \eqref{eq: modulation definition}, a system described by \eqref{eq: state eq}, self-modulation operator \eqref{eq: self modulation def}, and the linear map $\mathsf{R}{(\mathbf{u})}$ defined in \eqref{eq: H and R map definition}. Then, if the self-modulation
    \begin{equation}
        \mathsf{S}(\mathbf{v})=\langle \mathbf{v}, \mathbf{v}(t_1-\tau) \rangle,
    \end{equation}
    has full rank, then there exists an input $\mathbf{u}(t)$ given by \eqref{eq: input from modulation for natural trajectories}.
\end{prop}

\begin{pf}
    Starting from Proposition \ref{prop: trajectory constraint}, it is necessary and sufficient that $\text{range}(\mathsf{H}) \subseteq \text{range}(\mathsf{R})$, which is always true if $\text{range}(\mathsf{R})=\mathbb{R}^{n}$.
    
    Using the fact that $\text{range}(\mathsf{R})=\text{range}(\mathsf{R}\mathsf{R}^*)$ and computing $\mathsf{R}\mathsf{R}^*$ gives the $n \times n$ symmetric matrix
    \begin{equation}
        \mathsf{R}\mathsf{R}^*=\int_{t_0}^{t_1} \mathbf{v}^\top(t_1-\tau) \mathbf{v}(t_1-\tau) \dd \tau\text{,}
    \end{equation}
    which can be rewritten as
    \begin{equation}
        \mathsf{R}\mathsf{R}^*=\int_{0}^{T} \mathbf{v}^\top(\tau) \mathbf{v}(\tau) \dd \tau = \mathsf{S}(\mathbf{v}) \text{.}
    \end{equation}
    Thus, if $\rank ~\mathsf{S}(\mathbf{v})=n$, $\text{range}(\mathsf{R})=\mathbb{R}^{n}$. \hfill $\blacksquare$
\end{pf}

\begin{exmp}\label{example: reachability gramian}
    Consider the case $\bm{\uplambda}=\bm{0}$ for the auxiliary system \eqref{eq: aux system state eq}-\eqref{eq: aux system out eq}. Then, the modulation kernel $\mathbf{g}(\tau)$ leads to the state-transition matrix, as previously discussed, and by selecting $\mathbf{g}(0)=\mathbf{I}$ one obtains
    \begin{equation}
        \mathbf{v}^\top(t_1-\tau)= \bm{\Phi}(t_1, \tau) \mathbf{B}(\tau)
    \end{equation}
    from \eqref{eq: aux system state eq}-\eqref{eq: aux system out eq} and, subsequently,
    \begin{equation}\label{eq: traditional reachability gramian}
        \mathsf{S}(\mathbf{v})=\int_{t_0}^{t_1} \bm{\Phi}(t_1, \tau) \mathbf{B}(\tau) \mathbf{B}^\top(\tau) \bm{\Phi}^\top (t_1, \tau) \dd \tau \text{,}
    \end{equation}
    which is the well-known reachability gramian $\mathbf{W}_{R}(t_1, t_0)$ \citep[p.~176]{ChenLinearSystems}, a structure that has already been discussed using the range of linear maps in some of the literature \citep[p.~226]{Callier_linear_system_theory}. Thus, the reachability gramian can be seen as a particular case of the self-modulation $\mathsf{S}(\mathbf{v})$ for the dual modulation selection given by $\bm{\uplambda}(\tau)=\bm{0}$. \hfill $\diamond$
\end{exmp}

Example \ref{example: reachability gramian} shows a clear interpretation of the self-modulation of $\mathbf{v}$: it is a generalization of the reachability gramian, and can be directly understood as representing the range of the generalized reachability map $\mathsf{R}(\mathbf{u})$ defined in \eqref{eq: H and R map definition}.

Moreover, it is possible to extend the discussion in Example \ref{example: reachability gramian} to arbitrary dual modulations $\bm{\uplambda}$ through the self-modulation $\mathsf{S}(\mathbf{v})$ and auxiliary system \eqref{eq: aux system state eq}-\eqref{eq: aux system out eq}. Since it is well-known that controllability is a necessary and sufficient condition for the inputs to be able to interact and regulate all states \citep[pp.~176-177]{ChenLinearSystems}, the dual modulation $\bm{\uplambda}$ must, in some sense, preserve controllability, which can be evaluated using the proposition below.

\begin{prop}\label{prop: gamma feasible trajectory constraint}
    Consider modulation operator \eqref{eq: modulation definition}, a controllable system given by \eqref{eq: state eq}, a dual modulation \eqref{eq: dual modulation definition}, let $\overline{\mathbf{W}}_R(\tau)\in \mathbb{R}^{n \times n}$ be given as
    \begin{equation}\label{eq: W_R bar}
        \overline{\mathbf{W}}_R(\tau):= \bm{\Phi}(t_1, \tau) \mathbf{B}(\tau) \mathbf{B}^\top(\tau) \bm{\Phi}^\top (t_1, \tau),
    \end{equation}
    and let $\bm{ \Lambda}(\tau) \in \mathbb{R}^{n \times n}$ be given as \eqref{eq: Gamma of h definition}. Then, if the integral equation
    \begin{equation}\label{eq: feasible trajectories Gamma constraint}
         \langle \bm{ \Lambda}, \overline{\mathbf{W}}_R \rangle + \langle \bm{ \Lambda},  \overline{\mathbf{W}}_R \rangle ^\top+ \langle \bm{ \Lambda}, \overline{\mathbf{W}}_R \bm{ \Lambda}(t_1-\tau) \rangle =\bm{0},
    \end{equation}
    is satisfied, then there exists an input $\mathbf{u}(t)$ given by \eqref{eq: input from modulation for natural trajectories}.
\end{prop}
\begin{pf}
    Take $\mathbf{v}(\tau)$ given by \eqref{eq: aux system out eq} and compute its self-modulation to obtain
    \begin{equation}
            \mathsf{S}(\mathbf{v})=\langle \mathbf{B}^* \mathbf{g}, \mathbf{B}^*(t_1-\tau) \mathbf{g}(t_1-\tau) \rangle,
    \end{equation}
    which, by substituting \eqref{eq: solution of LTV auxiliary system with h using Gamma}, defining $\overline{\mathbf{W}}_R(\tau) \in \mathbb{R}^{n \times n}$ as \eqref{eq: W_R bar}, and reorganizing leads to
    \begin{equation}
    \begin{aligned}
        \mathsf{S}(\mathbf{v})=& \langle \bm{ \Lambda}, \overline{\mathbf{W}}_R \rangle + \langle \bm{ \Lambda},  \overline{\mathbf{W}}_R \rangle ^\top\\&+ \langle \bm{ \Lambda}, \overline{\mathbf{W}}_R \bm{ \Lambda}(t_1-\tau) \rangle+\int_{t_0}^{t_1}\overline{\mathbf{W}}_R(\tau) \dd \tau.
    \end{aligned}
    \end{equation}
    Lastly, for any controllable system, it follows that
    \begin{equation}
        \rank~ \left( \int_{t_0}^{t_1} \overline{ \mathbf{W} }_R(\tau) \dd \tau \right)=n.
    \end{equation}
    Thus, due to Proposition \ref{proposition: reachability gramian}, it suffices that \eqref{eq: feasible trajectories Gamma constraint} is satisfied for an input $\mathbf{u}(t)$ given by \eqref{eq: input from modulation for natural trajectories} to exist. \hfill $\blacksquare$
\end{pf}

This motivates the definition of a new operator to describe \eqref{eq: feasible trajectories Gamma constraint}, here denoted as the controllability bracket.

\begin{defn}
    Consider the modulation operator \eqref{eq: modulation definition}, a dual modulation \eqref{eq: dual modulation definition}, $\bm{ \Lambda}(\tau)$ as in \eqref{eq: Gamma of h definition}, and let $\overline{\mathbf{W}}_R$ be as in \eqref{eq: W_R bar}. Then, the controllability bracket is defined as
    \begin{equation}\label{eq: controllability bracket definition}
        \llangle \bm{ \Lambda}, \overline{\mathbf{W}}_R \rrangle:=\langle \bm{ \Lambda}, \overline{\mathbf{W}}_R \rangle + \langle \bm{ \Lambda},  \overline{\mathbf{W}}_R \rangle ^\top+ \langle \bm{ \Lambda}, \overline{\mathbf{W}}_R \bm{ \Lambda}(t_1-\tau) \rangle.
    \end{equation}
\end{defn}

In doing so, it is possible to rephrase Proposition \ref{prop: gamma feasible trajectory constraint} using the controllability bracket operator.

\begin{cor}
    Given a controllable system, if a dual modulation selection $\bm{\uplambda}(\tau)$ and the associated integral $\bm{ \Lambda}(\tau)$ satisfy
    \begin{equation}\label{eq: controllability bracket = 0}
        \llangle \bm{ \Lambda}, \overline{\mathbf{W}}_R \rrangle= \bm{0},
    \end{equation}
    then there exists $\mathbf{u}(t)$ given by \eqref{eq: input from modulation for natural trajectories}.
\end{cor}

\begin{exmp}\label{example: controllability bracket natural trajectories}
    Consider the selection given by $\bm{\uplambda}=\bm{0}$. Then, it trivially follows that $\bm{ \Lambda}(\tau)=\bm{0}$ and, consequently,
    \begin{equation}
        \llangle \bm{0}, \overline{\mathbf{W}}_R \rrangle =\bm{0},
    \end{equation}
    showing that this selection trivially preserves the controllability of the system. In fact, any dual modulation that satisfies $\bm{ \Lambda}(\tau)=\bm{0}$ preserves controllability, but they do not drive the system to the origin in a finite time since they do not satisfy \eqref{eq: Gamma(T)=-I}. \hfill $\diamond$
\end{exmp}

\begin{exmp}\label{example: controllability bracket example 2}
    Consider again the dual modulation from Example \ref{example: dirac delta without controllability bracket}, i.e. $\bm{\uplambda}(\tau)=-\bm{\Phi}^\top(t_1-\tau, t_1)\delta(\tau-t_1)$ with
    \begin{equation}
        \bm{ \Lambda}(\tau)=\left\{ \begin{array}{ccc}
              \bm{0} &, & \tau \in [0, T) \\
             -\bm{I} &, & \tau=T
        \end{array} \right. .
    \end{equation}
    Then, it directly follows that
    \begin{equation}
        \llangle \bm{ \Lambda}, \overline{\mathbf{W}}_R \rrangle =\bm{0}
    \end{equation}
    and, consequently, there exists $\mathbf{u}(t)$ given by \eqref{eq: input from modulation for natural trajectories} if the system is controllable. The reader is asked to compare the presented analysis to Example \ref{example: dirac delta without controllability bracket}. \hfill $\diamond$ 
\end{exmp}

\subsection{From the MFM to Fixed-Time Control}\label{subsec: from the MFM to fixed-time control}

Instead of selecting the dual modulation $\bm{\uplambda}$ and obtaining $\mathbf{g}$ as the solution to auxiliary system \eqref{eq: aux system state eq}, it is possible to directly select $\mathbf{g}$ and obtain $\bm{\uplambda}$ as a solution to \eqref{eq: aux system state eq}: a concept that was discussed in, e.g., Example 5 of \cite{Acc_DT_MFM} and in \cite{IonesiRJ19}. In this way, by simply constructing $\mathbf{g}(\tau)$ using TMFs and unitary RMFs, the boundary conditions $\mathbf{g}(0)=\mathbf{I}$ and $\mathbf{g}(T)=\bm{0}$ can be trivially satisfied, having $\mathbf{v}(\tau)$ directly obtained from \eqref{eq: aux system out eq} and allowing for the direct evaluation of $\mathsf{S}(\mathbf{v})$.

\begin{cor}\label{corollary: kernel selection as TMFs and unitary RMFs}
    Consider modulation operator \eqref{eq: modulation definition}, a controllable system given by \eqref{eq: state eq}, $\mathbf{v}(\tau)$ given by \eqref{eq: aux system out eq}, and let $T \in \mathbb{R}_{>0}$. Then, a selection $\mathbf{g}(\tau)$ given by
    \begin{equation}
        \mathrm{g}_{ij}(\tau)=\left\{ \begin{array}{lcr}
            \text{a unitary RMF of order } \geq 1 &, & i= j  \\
            \text{a TMF of order } \geq 1 &, & i\neq j   
        \end{array} \right.
    \end{equation}
    that satisfies $\rank~\mathsf{S}(\mathbf{v})=n$ for
    \begin{equation}
        \mathsf{S}(\mathbf{v})=\int_{t_0}^{t_1} \mathbf{g}^\top(t_1-\tau) \mathbf{B}(\tau) \mathbf{B}^\top(\tau) \mathbf{g}(t_1-\tau) \dd \tau
    \end{equation}
    guarantees the existence of an input $\mathbf{u}(t)$ given by \eqref{eq: input from modulation for natural trajectories} that drives the system to the origin in a finite-time.
\end{cor}
\begin{pf}
    This directly follows from the definition of TMFs and unitary RMFs, together with Proposition \ref{prop: Gamma(T)=-I}, and Proposition \ref{proposition: reachability gramian}.
\end{pf}

\begin{exmp}\label{example: self-modulation with LTI B matrix}
    As a simple example on evaluating the existence of inputs given by \eqref{eq: input from modulation for natural trajectories} using Corollary \ref{corollary: kernel selection as TMFs and unitary RMFs}, consider a simple time-invariant $\mathbf{B}$ matrix given by
    \begin{equation}
        \mathbf{B}=\begin{bmatrix}
            \frac{1}{0.5} \\ 0
        \end{bmatrix}
    \end{equation}
    and compute $\mathsf{S}(\mathbf{v})$ to obtain
    \begin{equation}
        \mathsf{S}(\mathbf{v})= \int_{t_0}^{t_1}\begin{bmatrix}
            \frac{\mathrm{g}_{11}^2(t_1-\tau)}{0.25} & \frac{\mathrm{g}_{11}(t_1-\tau) \mathrm{g}_{12}(t_1-\tau)}{0.25} \\
            \frac{\mathrm{g}_{11}(t_1-\tau) \mathrm{g}_{12}(t_1-\tau)}{0.25} & \frac{\mathrm{g}_{12}^2(t_1-\tau)}{0.25} 
        \end{bmatrix} \dd \tau.
    \end{equation}
    Next, select
    \begin{gather}
        \mathrm{g}_{11}(\tau)=1-\tau/T \label{eq: RMF 11}, \\
        \mathrm{g}_{12}(\tau)=(e^{2\tau}-1)(e^{3\tau}-e^{3T}), \label{eq: TMF 12}
    \end{gather}
    and $T=1$. Finally, evaluate the integrals to obtain
    \begin{equation}
        \mathsf{S}(\mathbf{v})\approx\begin{bmatrix}
            1.33 &  & -32.35 \\
            -32.35 & & 2.33\cdot 10^3
        \end{bmatrix},
    \end{equation}
    clearly showing that $\rank~\mathsf{S}(\mathbf{v})=2$, meaning Corollary \ref{corollary: kernel selection as TMFs and unitary RMFs} is satisfied and there exists an input given by \eqref{eq: input from modulation for natural trajectories}. \hfill $\diamond$
\end{exmp}

Even though a rather large number of modulating functions may be necessary to construct $\mathbf{g}$ according to Corollary \ref{corollary: kernel selection as TMFs and unitary RMFs}, it was shown in \cite{Acc_Algebra_MF} that existing MFs can be combined to construct infinitely many new modulating functions in a structured way. Therefore, it is possible to specify the kernel $\mathbf{g}$ as a matrix of modulating functions for any finite $n$, leading to the following novel control law for LTV systems, denoted as modulating function regulator (MFR), which is directly obtained from the proposed framework.

\begin{thm}\label{theorem: MF-based MFR}
    Consider modulation operator \eqref{eq: modulation definition}, a controllable system given by \eqref{eq: state eq} satisfying $\rank~\mathbf{B}(t)=m$ for all $t$, and auxiliary system \eqref{eq: aux system state eq}-\eqref{eq: aux system out eq}. Then, by selecting $\mathbf{g}(\tau) \in \mathbb{R}^{n \times n}$ according to Corollary \ref{corollary: kernel selection as TMFs and unitary RMFs} and selecting $t_0=t-T$ and $t_1=t$, a moving window fixed-time modulating function regulator can be approximated using the fixed-step trapezoidal rule as
    \begin{equation}\label{eq: fixed-time control law}
        \mathbf{u}(t)=\mathbf{B}^\dagger(t) \bigl( \bm{\uplambda}^\top(0)\mathbf{x}(t) + \bm{\uplambda}^\top(T) \mathbf{x}(t-T)+2\mathbf{F} \bigr),
    \end{equation}
    where $\mathbf{B}^\dagger(t)$ is the pseudoinverse of $\mathbf{B}(t)$ defined as
    \begin{equation}\label{eq: pseudoinverse of B definition}
        \mathbf{B}^\dagger(t):=\left(\mathbf{B}^\top(t) \mathbf{B}(t)\right)^{-1} \mathbf{B}^\top,
    \end{equation}
    and
    \begin{equation}\label{eq: fixed_time control law F terms}
        \mathbf{F}:=\sum_{k=1}^{N-1} \begin{array}{l}
             \bm{\uplambda}^\top(T-k\Delta t) \mathbf{x}(t-T+k\Delta t)\\
             - \mathbf{v}^\top(T-k\Delta t) \mathbf{u}(t-T+k\Delta t)
        \end{array} ,
    \end{equation}
    with $N \in \mathbb{N}_{>0}$ partitions and step-size $\Delta t \in \mathbb{R}_{>0}$.
\end{thm}

\begin{pf}
    First, note that if system \eqref{eq: state eq} is controllable and the kernel $\mathbf{g}(\tau) \in \mathbb{R}^{n \times n}$ is selected according to Corollary \ref{corollary: kernel selection as TMFs and unitary RMFs}, then there exists an input given by \eqref{eq: input from modulation for natural trajectories} that drives the system from $\mathbf{x}(t_0)$ to $\mathbf{x}(t_1)=\bm{0}$ .\\
    Next, select $t_0=t-T$ and $t_1=t$ and consider a fixed-step trapezoidal rule to numerically approximate the integrals, namely
    \begin{equation}
    \begin{aligned}
        \frac{\Delta t}{2} \sum_{k=1}^{N} \mathbf{v}^\top(T-k\Delta t) \mathbf{u}(t_0+k\Delta t) \\+\frac{\Delta t}{2} \sum_{k=1}^{N} \mathbf{v}^\top(T-(k-1)\Delta t) \mathbf{u}(t_0+(k-1)\Delta t) =\\
        \frac{\Delta t}{2} \sum_{k=1}^{N} \bm{\uplambda}^\top(T-k\Delta t) \mathbf{x}(t_0+(k-1)\Delta t) \\+ \frac{\Delta t}{2} \sum_{k=1}^{N} \bm{\uplambda}^\top(T-(k-1)\Delta t) \mathbf{x}(t_0+(k-1)\Delta t)
    \end{aligned}
    \end{equation}
    where $N \in \mathbb{N}_{>0}$ is the number of partitions and $\Delta t \in \mathbb{R}_{>0}$ is the step-size. Then, use the fact that $\mathbf{v}(0)=\mathbf{B}(t)$ for $\mathbf{g}(0)=\mathbf{I}$ and $\mathbf{v}(T)=\bm{0}$ for $\mathbf{g}(T)=\bm{0}$ to obtain
    \begin{equation}
    \begin{aligned}
        \mathbf{B}(t)\mathbf{u}(t)+2\sum_{k=1}^{N-1} \mathbf{v}^\top(T-k\Delta t) \mathbf{u}(t_0+k\Delta t) = \\
        \sum_{k=1}^{N} \bm{\uplambda}^\top(T-k\Delta t) \mathbf{x}(t_0+(k-1)\Delta t) \\+ \sum_{k=1}^{N} \bm{\uplambda}^\top(T-(k-1)\Delta t) \mathbf{x}(t_0+(k-1)\Delta t)
    \end{aligned}        
    \end{equation}
    Finally, assuming that $\rank~\mathbf{B}(t)=m$, reorganizing, and using the pseudoinverse $\mathbf{B}^\dagger$ given by \eqref{eq: pseudoinverse of B definition}, the control law given by \eqref{eq: fixed-time control law}-\eqref{eq: fixed_time control law F terms} is obtained.    \hfill $\blacksquare$
\end{pf}

\begin{rem}
    Due to the direct relation between the pseudoinverse operator and least-squares regression \citep[p.~83]{Luenberger_Optimization_Vector_Spaces}, inputs obtained from control law \eqref{eq: fixed-time control law}-\eqref{eq: fixed_time control law F terms} are optimal in the least-squares-sense with respect to $\mathbf{B}$ and, consequently, will display a relatively small norm. This is also clearly seen in the numerical simulations shown in Section \ref{sec: numerical simulations}.
\end{rem}

\begin{exmp}
    Although the condition $\rank~\mathbf{B}(t)=m$ might seem restrictive, it suffices to take any two linearly dependent inputs $\mathrm{u}_1$ and $\mathrm{u}_2$ and define a virtual input $v$ as their linear combination. As a simple example, consider a two-input system with
    \begin{equation}
        \mathbf{B}=\begin{bmatrix}
            0 & & 0 \\ 1&  & 3t 
        \end{bmatrix},
    \end{equation}
    which does not satisfy $\rank~\mathbf{B}(t)=m$. By simply defining the virtual input $v=\mathrm{u}_1+3t\mathrm{u}_2$, it follows that the resulting $\mathbf{B}$ matrix is $[0, 1]^\top$ and has rank $1$, allowing for the computation of $v$. The mapping from $v$ to $\mathrm{u}_1$ and $\mathrm{u}_2$ can then be done using, e.g., a minimal cost selection. \hfill $\diamond$
\end{exmp}

Note that the control law given in Theorem \ref{theorem: MF-based MFR} explicitly contains the system model, allowing for a simple interpretation of the impact of parameter changes. In addition, the only difference between LTI and LTV systems is that, while the controller gains for LTI systems can be pre-computed and stored in memory, the controller gains for LTV systems must be computed online.

More importantly, it is clear to see that the system trajectory is directly influenced by the selected modulating functions when Proposition \ref{prop: trajectory constraint} and Corollary \ref{corollary: kernel selection as TMFs and unitary RMFs} are employed. In particular during the period $t \in [0, T)$, i.e. while the integration window is still being filled, the integral equation \eqref{eq: input from modulation for natural trajectories} becomes the Volterra equation of the first kind given by
\begin{equation}\label{eq: transient volterra equation}
    \int_{0}^{t} \bm{\uplambda}^\top(t-\tau) \mathbf{x}(\tau) \dd \tau = \int_{0}^t \mathbf{v}^\top (t-\tau) \mathbf{u}(\tau) \dd \tau.
\end{equation}

Although not necessary to obtain fixed-time convergence, how well the controller approximates a solution to \eqref{eq: transient volterra equation} directly dictates the transient behavior of the system before stabilization, which follows directly from Proposition \ref{prop: input modulation equal state modulation} and $\mathbf{g}(0)=\mathbf{I}$, i.e.
\begin{equation}
    \mathbf{x}(t)\hspace{-0.05cm}=\hspace{-0.05cm}\mathbf{g}^\top(t) \mathbf{x}(t_0) + \int_{0}^t \Bigl(\mathbf{v}^\top\hspace{-0.075cm} (t-\tau) \mathbf{u}(\tau)- \bm{\uplambda}^\top\hspace{-0.075cm}(t-\tau) \mathbf{x}(\tau)\Bigr) \dd \tau.
\end{equation}

Thus, although not necessarily matching this trajectory, the transient behavior of the system is explicitly influenced by the choice of modulating functions, and the trajectory in situations where \eqref{eq: transient volterra equation} is satisfied is given by
\begin{equation}\label{eq: transient period MFR}
    \mathbf{x}_{\Psi}(t)=\bm{\Psi}(t, 0)\mathbf{x}(0)=\mathbf{g}^\top(t) \mathbf{x}(0).
\end{equation}
Hence, it is possible to obtain information about how well the control law approximates a solution, if it exists, to the Volterra equation \eqref{eq: transient volterra equation} by simply comparing the system behavior to \eqref{eq: transient period MFR}. 

%% file: 05_Numerical_Examples.tex
\section{Numerical Simulations}\label{sec: numerical simulations}
To illustrate how the proposed approach can be used, numerical simulations of Theorem \ref{theorem: MF-based MFR} are shown below.

\begin{figure}[t]
  \centering
  \includegraphics[width=.49\linewidth]{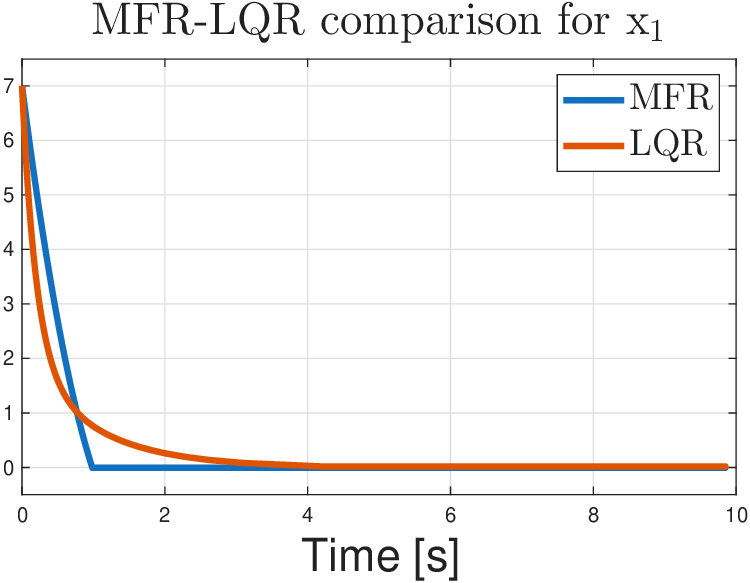}
  \makebox[0pt][r]{%
    \raisebox{3.5em}{%
      \includegraphics[width=.13\linewidth]{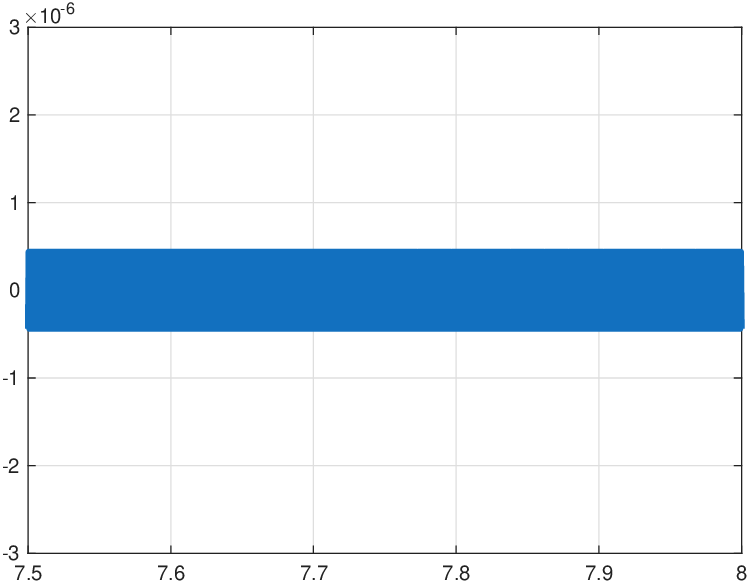}%
    }\hspace*{1.5em}%
  }%
  \includegraphics[width=.49\linewidth]{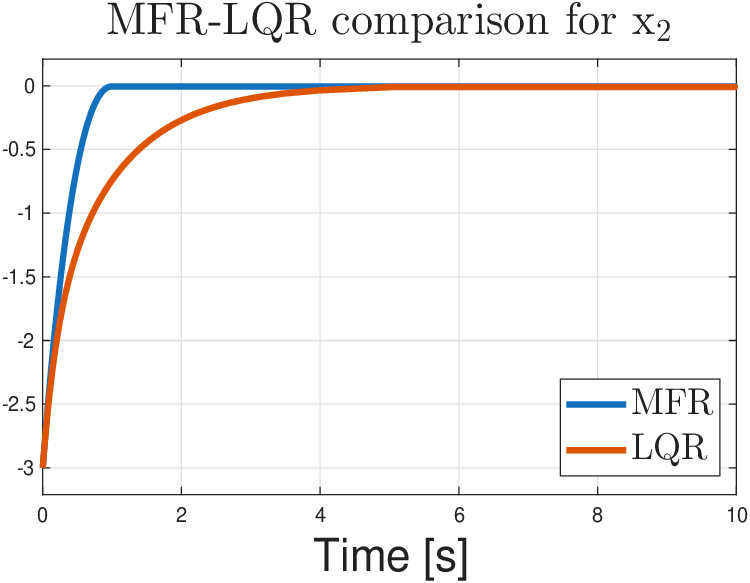}%
  \makebox[0pt][r]{%
    \raisebox{3.5em}{%
      \includegraphics[width=.13\linewidth]{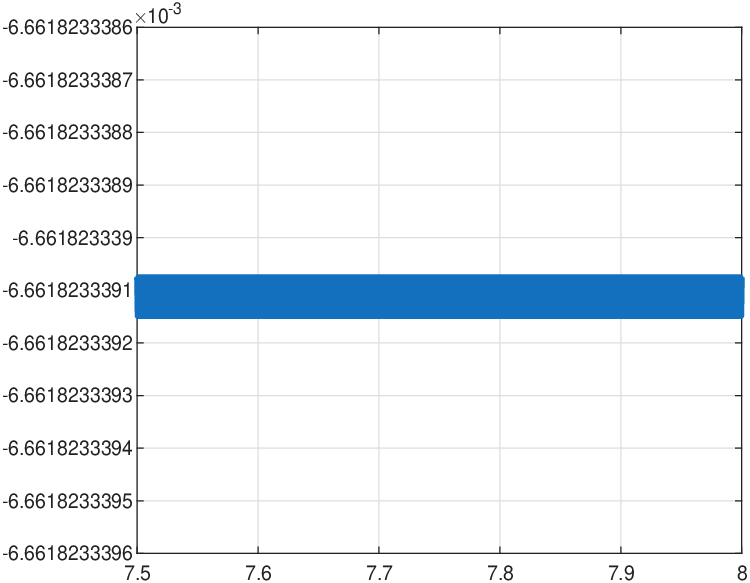}%
    }\hspace*{1.5em}%
  }%
  \caption{Comparison between LQR and MFR for \eqref{eq: simple unstable system}, with $T=1s$ and $N=2000$. Zoomed-in intervals are between $7.5s$ and $8s$.}
  \label{fig: simple system x1 and x2 comparison}
\end{figure}

\subsection{Unstable LTI system}
Consider an unstable system given by
\begin{equation}\label{eq: simple unstable system}
    \dot{\mathbf{x}}=\begin{bmatrix}
        4 & & -3 \\ 1 & & 0
    \end{bmatrix} \mathbf{x} + \begin{bmatrix}
        2 \\ 0
    \end{bmatrix} \mathbf{u},
\end{equation}
which has poles at $3$ and $1$, and consider the initial conditions $\mathbf{x}(0)=[7, -3]^\top$. Next, select the modulating functions
\begin{gather}
    \mathrm{g}_{11}(\tau)= \left( 1-\frac{\tau}{T}\right)\label{eq: g11 selection for simple LTI system}\\
    \mathrm{g}_{12}(\tau)=\left(\frac{\tau}{T}\right)^3 \left( 1-\frac{\tau}{T}\right)^4\label{eq: g12 selection for simple LTI system}\\
    \mathrm{g}_{21}(\tau)= \left(\frac{\tau}{T}\right) \left( 1-\frac{\tau}{T}\right)\label{eq: g21 selection for simple LTI system}\\
    \mathrm{g}_{22}(\tau)= \left( 1-\frac{\tau}{T}\right)^2, \label{eq: g22 selection for simple LTI system}
\end{gather}
together with $T=1s$. Then, by direct computation of $\mathsf{S}(\mathbf{v})$, it follows that
\begin{equation}
    \mathsf{S}(\mathbf{v})\approx\begin{bmatrix}
        1.3333 & 0.0079 \\
        0.0079 & 0.0001
    \end{bmatrix},
\end{equation}
which is clearly rank 2.

By computing the derivatives of these modulating functions and sampling them, the control law \eqref{eq: fixed-time control law}-\eqref{eq: fixed_time control law F terms} can be computed in real-time, which was implemented assuming $\mathbf{x}(t)=\bm{0}$ and $\mathbf{u}(t)=\bm{0}$ for $t<0$. The number of partitions was selected as $N=2000$, i.e. $\Delta t=5 \cdot 10^{-4}s$, and a discretized LQR controller with tuning parameters
\begin{equation}
    \mathbf{Q}=\begin{bmatrix}
        \frac{1}{25} & & 0 \\
        0 & & \frac{1}{25}
    \end{bmatrix}, \; \; \; \mathrm{R}=\frac{1}{100},
\end{equation}
and the same step-size was used as a benchmark. The results are shown in Figure \ref{fig: simple system x1 and x2 comparison} and Figure \ref{fig: actuation use and interation error}.

\begin{figure}[t]
  \centering
  \includegraphics[width=.49\linewidth]{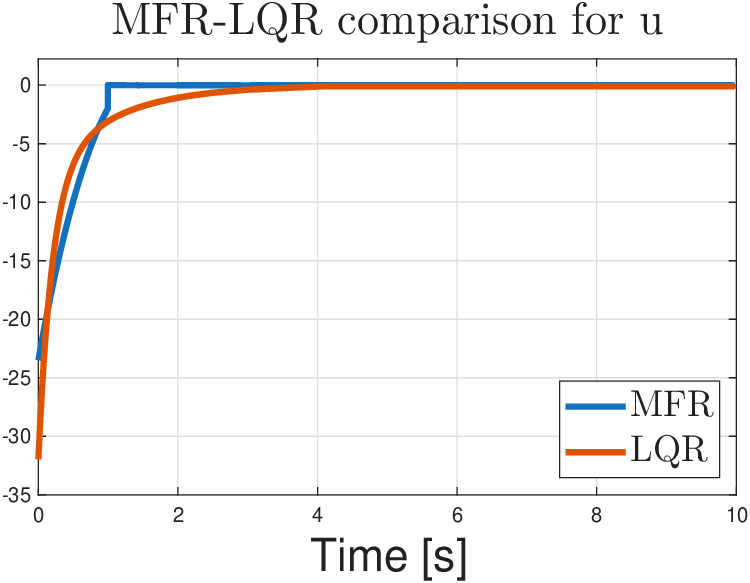}
  \makebox[0pt][r]{%
    \raisebox{3.5em}{%
      \includegraphics[width=.13\linewidth]{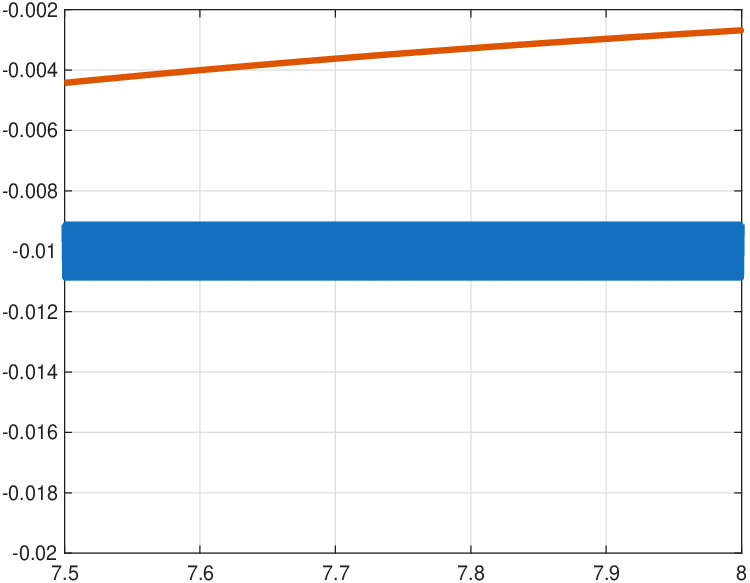}
    }\hspace*{1.5em}%
  }%
  \includegraphics[width=.49\linewidth]{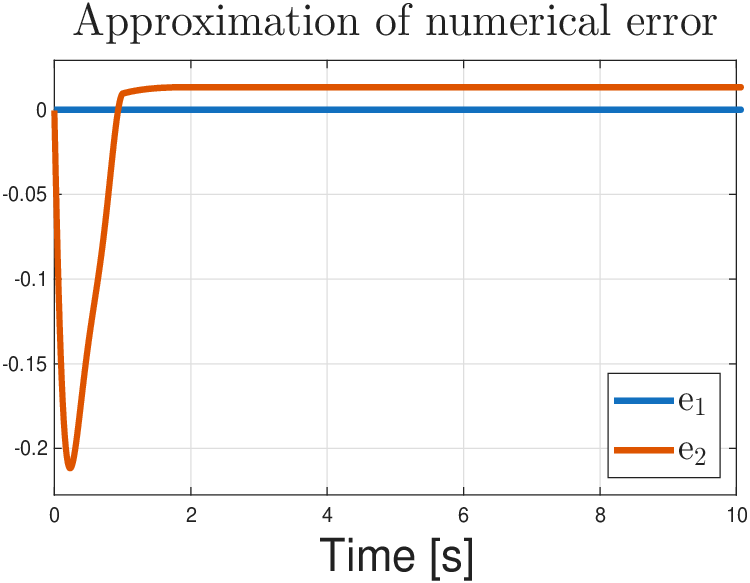}
  \makebox[0pt][r]{
    \raisebox{3.5em}{%
      \includegraphics[width=.13\linewidth]{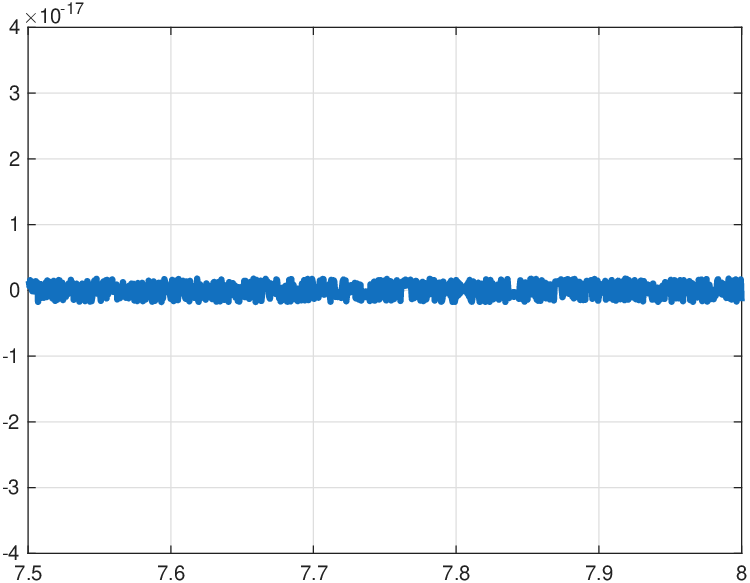}
    }\hspace*{0.5em}%
    \raisebox{3.5em}{%
      \includegraphics[width=.13\linewidth]{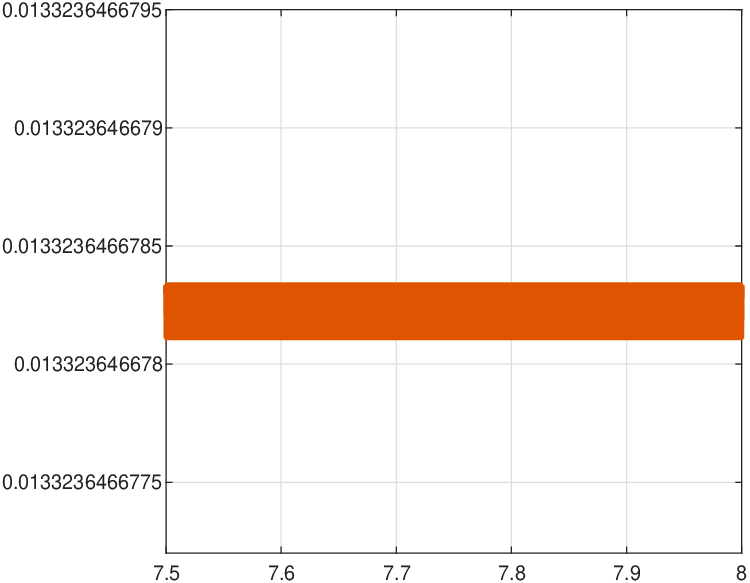}
    }\hspace*{0.5em}%
  }%
  \caption{Comparison between LQR and MFR control for \eqref{eq: simple unstable system}, with $T=1s$ and $N=2000$. Zoomed-in intervals are between $7.5s$ and $8s$.}
  \label{fig: actuation use and interation error}
\end{figure}

It is clear to seen in Figure \ref{fig: simple system x1 and x2 comparison} that, not only did the proposed approach stabilize the system, but it indeed did so in a fixed time, without needing to discuss the eigenvalues of the closed-loop system. Moreover, the RMS value of the state $\mathrm{x}_1$ with an LQR controller was $0.8857$, while the RMS for $\mathrm{x}_1$ with the MFR was $1.1455$: an increase of about $29\%$ with respect to the LQR value. In contrast, the RMS value obtained for $\mathrm{x}_2$ with the MFR was $0.4017$, while for the LQR was $0.5262$, an increase of about $31\%$ with respect to the MFR value. By combining the RMS value for both states in a vector, one can also compare their norms, obtained as $1.0302$ for the LQR, and $1.2139$ for the MFR.

Another point of comparison is the actuation use, seen in Figure \ref{fig: actuation use and interation error}. Both LQR and MFR input signals are within the same range, and also display similar RMS values: $3.9191$ for the LQR and $3.9613$ for the MFR. This indicates that, despite the fact that the modulating functions were arbitrarily chosen, the obtained performance using the MFR is comparable to the one obtained with the LQR, as also seen in Table \ref{tab: unstable system example RMS comparison}.

\begin{table}[h]
    \centering
    \caption{Comparison between RMS values obtained using LQR and MFR for \eqref{eq: simple unstable system}.}
    \label{tab: unstable system example RMS comparison}
    \begin{tabular}{|l|c|c|c|c|}
        \hline
         \textit{RMS} & $\mathrm{x}_1$ & $\mathrm{x}_2$ & $||\mathbf{x}_{\text{RMS}}||$ & $\mathrm{u}$ \\
         \hline
        LQR & 0.8857 & 0.5262 & 1.0302 & 3.9191 \\
        \hline
        MFR & 1.1455 & 0.4017 & 1.2139 & 3.9613\\
        \hline
    \end{tabular}    
\end{table}

Note, however, that the states do not converge exactly to $0$, but chatter around a stable value, as shown in the zoomed-in segments in Figure \ref{fig: simple system x1 and x2 comparison}. These rapid oscillations, and bias induction, can be attributed to errors associated with numerically implementing \eqref{eq: input from modulation for natural trajectories} using Theorem \ref{theorem: MF-based MFR}. A similar effect is seen in the generated input signal, shown on the left-side of Figure \ref{fig: actuation use and interation error}.

\begin{figure}[t]
    \centering
    \includegraphics[width=0.49\linewidth]{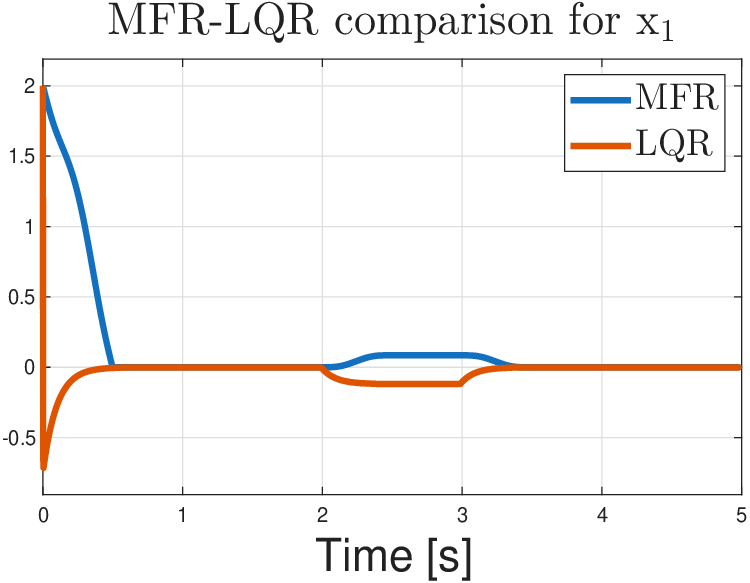}
    \includegraphics[width=0.49\linewidth]{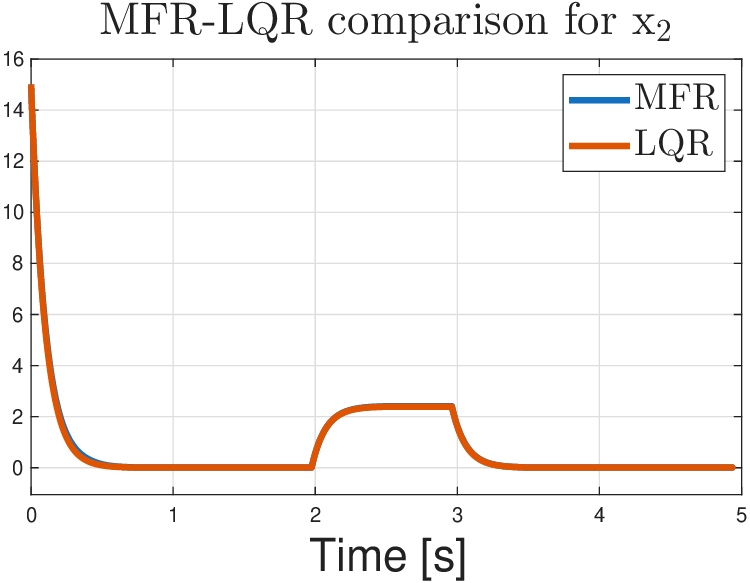}
    \caption{Comparison between state trajectories obtained using LQR and MFR, with $T=0.5$ and $N=2000$, for \eqref{eq: DC motor model}. Note that, even though the convergence time was expected to be $T=0.5s$, it took about $0.7s$ for the system to converge.}
    \label{fig: LTV DC motor state comparison}
\end{figure}

\begin{figure}[b]
    \centering
    \includegraphics[width=0.49\linewidth]{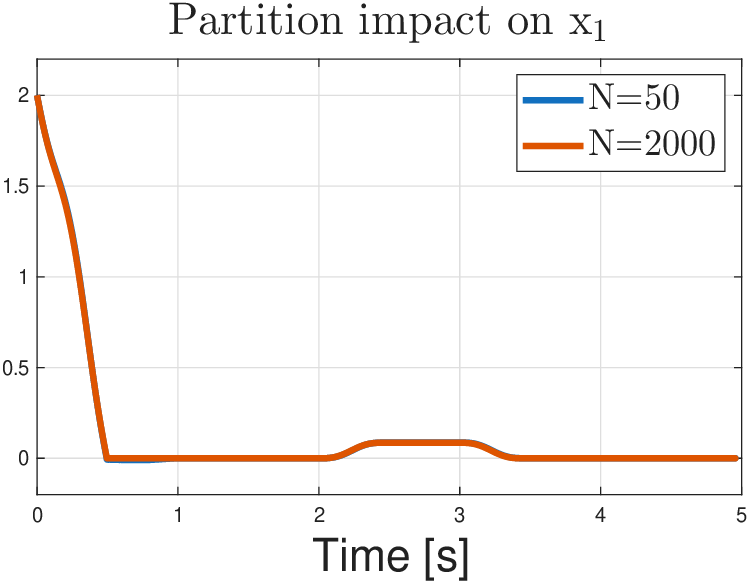}
    \includegraphics[width=0.49\linewidth]{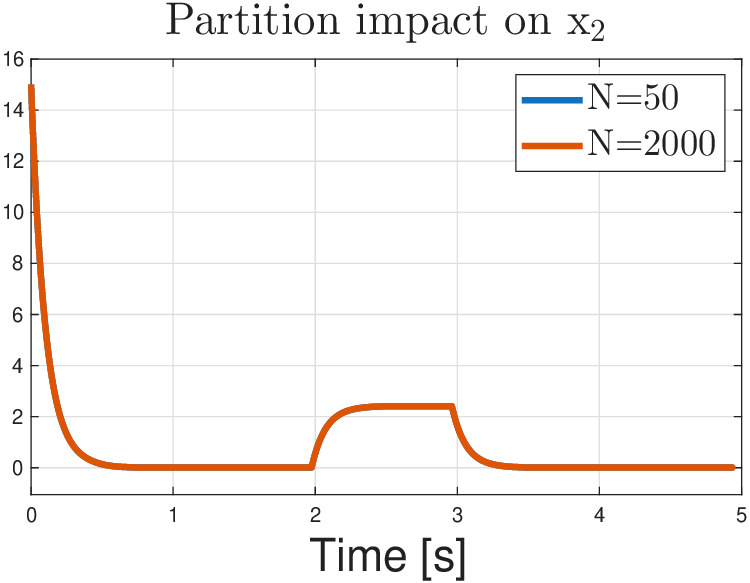}
    \caption{State trajectory comparison for \eqref{eq: DC motor model} using an MFR with $N=2000$ and $N=50$, where $T=0.5s$.}
    \label{fig: LTV DC motor state comparison N=50}
\end{figure}

Nonetheless, by computing $\mathbf{e}(t):=\langle \bm{\uplambda}, \mathbf{x} \rangle -\langle \mathbf{v}, \mathbf{u} \rangle$, it is possible to obtain an estimate for the numerical errors introduced, as seen on the right-side of Figure \ref{fig: actuation use and interation error}. Even though the numerical error is close to zero after the integration window is full, it does not become identically zero, and is not exactly centered around $0$ for $\mathrm{e}_2$, indicating a bias.

\subsection{DC Motor with Time-Varying Resistance}

Consider now a DC motor with a time-varying resistance and an external unmatched disturbance, given by
\begin{equation}\label{eq: DC motor model}
    \dot{\mathbf{x}}=\begin{bmatrix}
        -\frac{R(t)}{L} & -\frac{k_b}{L} \\
        \frac{k_t}{J} & -\frac{b}{J}
    \end{bmatrix} \mathbf{x} + \begin{bmatrix}
        \frac{1}{L} \\ 0
    \end{bmatrix} \mathrm{u} + \begin{bmatrix}
        0 \\ 1
    \end{bmatrix} d
\end{equation}
where $R(t)=1+t/2$, $L=0.5$, $k_t=k_b=0.02$, $b=0.01$, and $J=0.01$, having the disturbance $d$ being a step of amplitude $24$, which is active for $t \in [2, 3]$. Moreover, the initial conditions were selected as $\mathbf{x}(0)=[2, 15]^\top$.

Next, select the modulating functions $\mathrm{g}_{11}(\tau)$ as \eqref{eq: RMF 11}, 
$\mathrm{g}_{12}(\tau)$ as \eqref{eq: TMF 12},
\begin{gather}
    \mathrm{g}_{21}(\tau)=\left(\frac{\tau}{T}\right)^3 \left( 1-\frac{\tau}{T}\right)^3,\\
    \mathrm{g}_{22}(\tau)=\sech(3\tau) \frac{\tanh(\tau-T)}{\tanh(-T)},
\end{gather}
along with $T=0.5s$, $N=2000$, $\Delta t = 2.5 \cdot 10^{-4}s$, and $\mathbf{x}(t)=\bm{0}$ and $\mathbf{u}(t)=\bm{0}$ for $t<0$. As in the previous example, a discretized LQR with tuning parameters
\begin{equation}
    \mathbf{Q}=\begin{bmatrix}
        \frac{1}{10} & & 0 \\
        0 & & \frac{1}{10}
    \end{bmatrix}, \; \; \; \mathrm{R}=\frac{10^{-4}}{(24)^2},
\end{equation}
and the same step-size was used as a baseline comparison.

By direct computation, it follows that
\begin{equation}
    \mathsf{S}(\mathbf{v})\approx\begin{bmatrix}
        0.6687 & & -0.1879 \\
        -0.1879 & & 0.1123
    \end{bmatrix},
\end{equation}
which, again, clearly has rank 2.

\begin{figure}[t]
    \centering
    \includegraphics[width=0.49\linewidth]{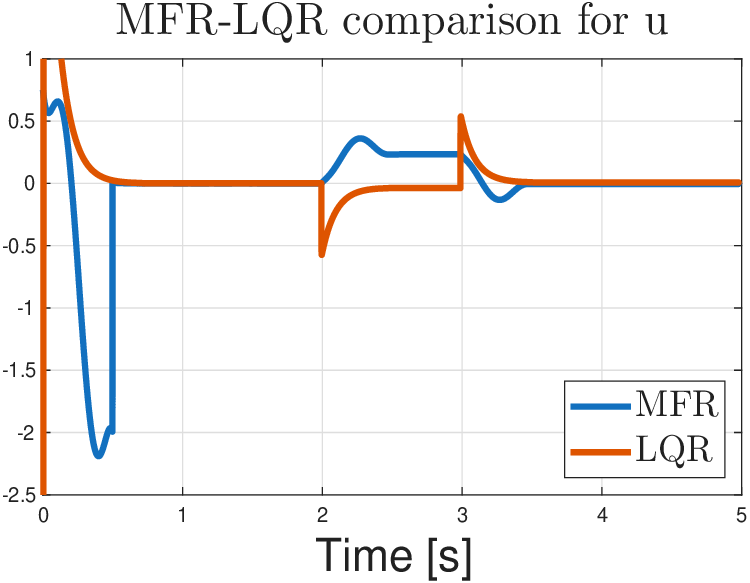}
    \includegraphics[width=0.49\linewidth]{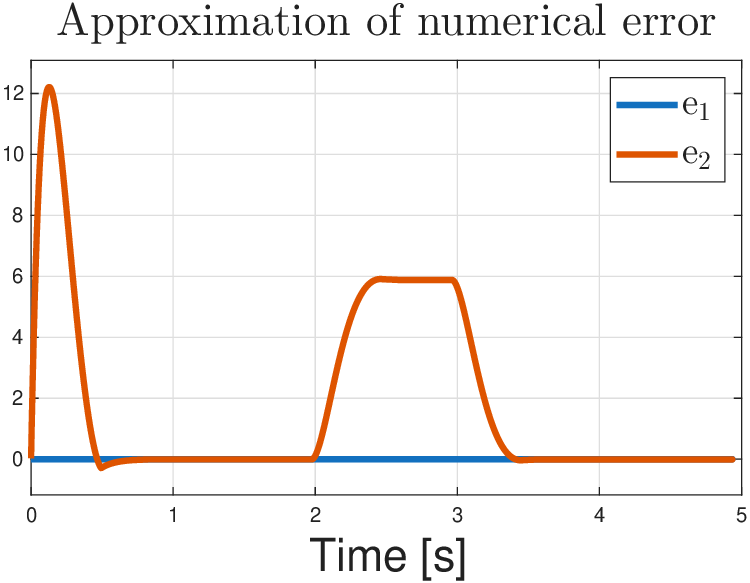}
    \caption{Comparison between actuation use (left) obtained using LQR and MFR, with $T=0.5$ and $N=2000$, for \eqref{eq: DC motor model}; and an approximation of the numerical errors (right). Note that the numerical error only converged at $t\approx 0.7s$.}
    \label{fig: LTV DC motor input and integration error}
\end{figure}

\begin{figure}[b]
    \centering
    \includegraphics[width=0.49\linewidth]{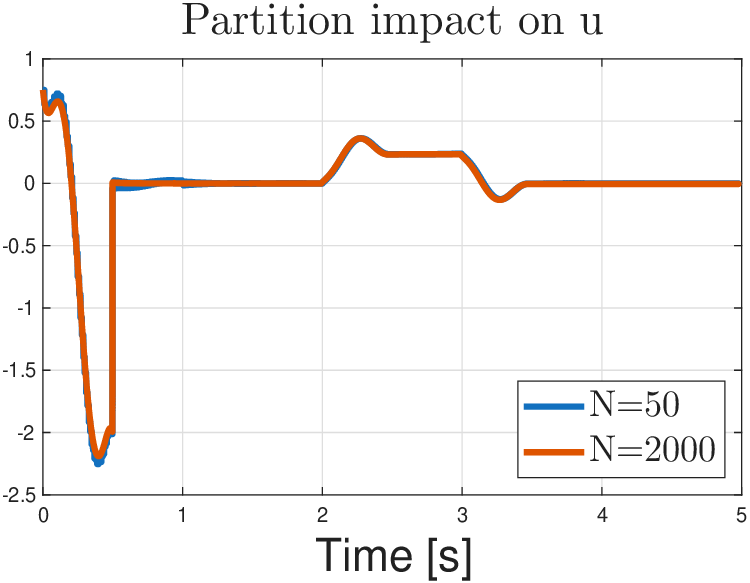}
    \includegraphics[width=0.49\linewidth]{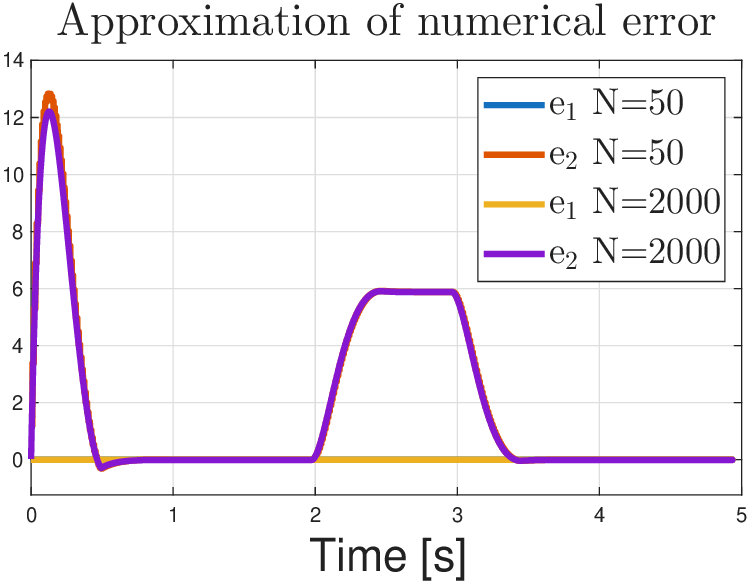}
    \caption{Actuation use and integration error comparison for \eqref{eq: DC motor model} using an MFR with $T=0.5s$, $N=2000$, and $N=50$.}
    \label{fig: LTV DC motor input and integration error N=50}
\end{figure}

The simulation results are shown in Figure \ref{fig: LTV DC motor state comparison}, where it is visible that the performance of the MFR controller is comparable to the benchmark LQR: in fact, the $\mathrm{x}_2$ trajectories almost perfectly overlap.

When discussing the disturbance rejection within the period between $2s$ and $3.5s$, the peak for $\mathrm{x}_1$ using LQR was around $-0.11$, while for the MFR was around $0.086$. As for $\mathrm{x}_2$, both peaks are around $2.4$, indicating similar levels of disturbance rejection capabilities.

\begin{figure}[t]
    \centering
    \includegraphics[width=0.49\linewidth]{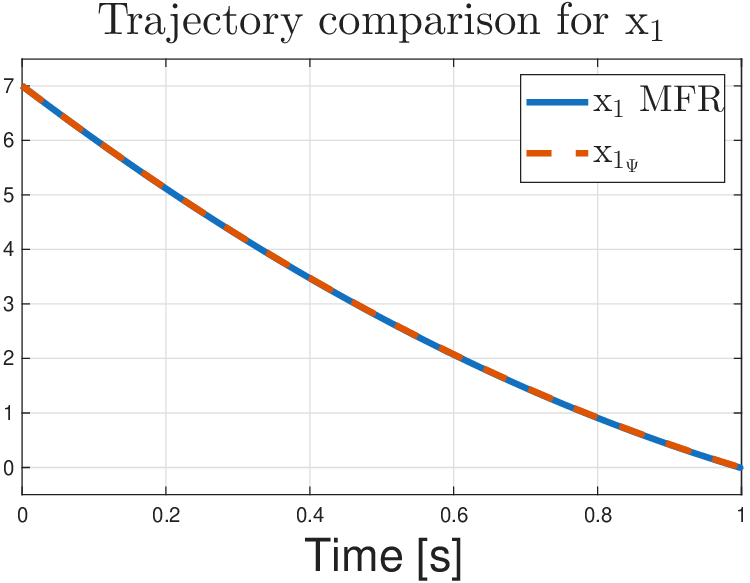}
    \includegraphics[width=0.49\linewidth]{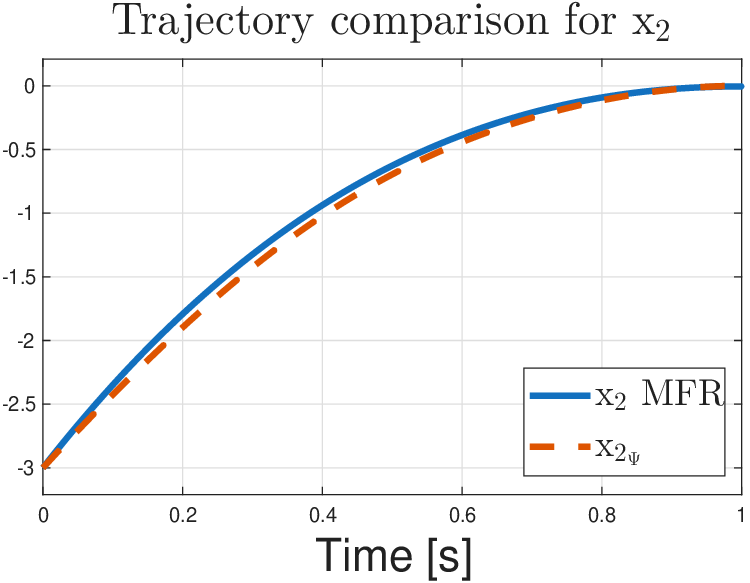}
    \caption{Comparison of the transient behavior of \eqref{eq: simple unstable system} with an MFR to the trajectory $\mathbf{x}_{\Psi}$ given by \eqref{eq: transient period MFR}.}
    \label{fig: trajectory comparison unstable LTI example}
\end{figure}

Regarding the actuation use, shown in Figure \ref{fig: LTV DC motor input and integration error} on the left, the MFR controller effectively avoided sharp peaks within the transient period, with a striking RMS value of $0.4476$, while the input obtained from the LQR had an RMS value of $20.1392$. Similar to the previous example, the RMS values are shown in Table \ref{tab: LTV DC motor example RMS comparison} for an easier comparison, and it is clear to see that, despite the significantly lower actuation use in the MFR than in the LQR, the norm of the RMS states was still similar.

\begin{table}[h]
    \centering
    \caption{Comparison between RMS values obtained using LQR and MFR for \eqref{eq: DC motor model}.}
    \label{tab: LTV DC motor example RMS comparison}
    \begin{tabular}{|l|c|c|c|c|}
        \hline
         \textit{RMS} & $\mathrm{x}_1$ & $\mathrm{x}_2$ & $||\mathbf{x}_{\text{RMS}}||$ & $\mathrm{u}$ \\
         \hline
        LQR & 0.0904 & 1.8086 & 1.8108 & 16.6068 \\
        \hline
        MFR & 0.3976 & 1.8299 & 1.8726 & 0.4476\\
        \hline
    \end{tabular}    
\end{table}

Nonetheless, it is also possible to see in Figure \ref{fig: LTV DC motor state comparison} that the system takes slightly longer than $T=0.5s$ to converge, in particular for $\mathrm{x}_2$. This can again be associated with approximation errors in the numerical solution, as seen in the right part of Figure \ref{fig: LTV DC motor input and integration error}. 

In addition, one can evaluate the impact of the selection of number of partitions during numerical integration. Instead of using $N=2000$, i.e. $\Delta t=2.5 \cdot 10^{-4}s$, consider the case where $N=50$, i.e. $\Delta t=0.01s$, shown in Figures \ref{fig: LTV DC motor state comparison N=50} and \ref{fig: LTV DC motor input and integration error N=50}. Despite the drastic decrease in the number of partitions, the integration error and state trajectories remained similar, but the chattering became visible in some segments of the input, such as around $t=0.5s$. This indicates that, although the proposed approach is influenced by the number of partitions, it does not seem to be extremely sensitive to it.

\subsection{Transient Behavior}

Lastly, it is also a topic of interest to investigate the transient trajectory. As discussed in the text around \eqref{eq: transient volterra equation}-\eqref{eq: transient period MFR}, the selection of modulating functions directly influences the transient behavior of the system, a fact that is evident during the numerical simulations. 

In the first example, i.e. \eqref{eq: simple unstable system}, the trajectory given by \eqref{eq: transient period MFR} is almost perfectly followed for both $\mathrm{x}_1$ and $\mathrm{x}_2$, as seen in Figure \ref{fig: trajectory comparison unstable LTI example}, clearly showing how the selection of MFs influences the transient behavior. Moreover, by comparing the period $[0, 0.5]$ in Figure \ref{fig: actuation use and interation error} with Figure \ref{fig: trajectory comparison unstable LTI example}, it is clear to see how the approximated numerical error correlates with the distance to the trajectory $\mathbf{x}_{\Psi}$: a negative error indicates that the current trajectory is above $\mathbf{x}_{\Psi}$ and a positive error shows that it is below it.

\begin{figure}[t]
    \centering
    \includegraphics[width=0.49\linewidth]{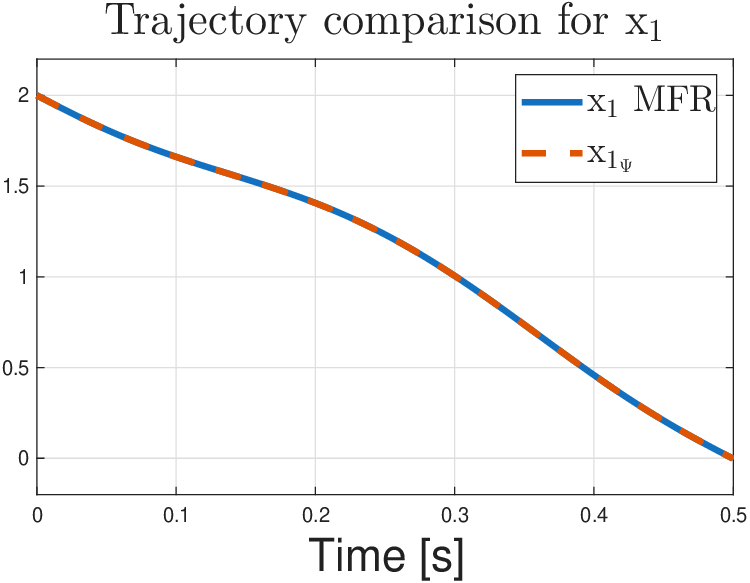}
    \includegraphics[width=0.49\linewidth]{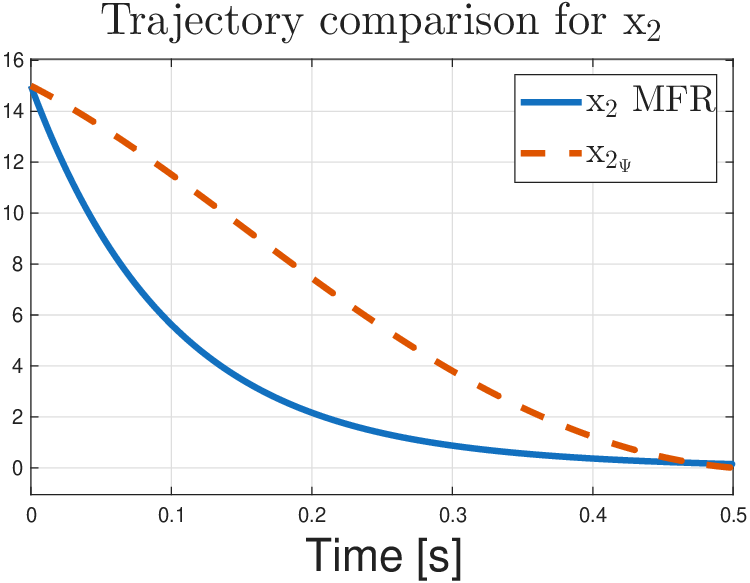}
    \caption{Comparison of the transient behavior of \eqref{eq: DC motor model} with an MFR to the trajectory $\mathbf{x}_{\Psi}$ given by \eqref{eq: transient period MFR}.}
    \label{fig: trajectory comparison LTV DC motor}
\end{figure}

A similar comparison can be made with the results related to \eqref{eq: DC motor model}, and they are shown in Figure \ref{fig: trajectory comparison LTV DC motor}. Although the trajectories almost perfectly aligns for $\mathrm{x}_1$, the same cannot be said about $\mathrm{x}_2$, which remained below the ideal trajectory for almost the entire window. This is also visible in Figure \ref{fig: LTV DC motor input and integration error}, where a large positive error was obtained from the estimate $\mathrm{e}_2$ within $[0, T]$, indicating that the current trajectory for $\mathrm{x}_2$ was below $\mathrm{x}_{2_\Psi}$.

Nevertheless, it is important to note that, even in situations where the Volterra equation \eqref{eq: transient volterra equation} is not satisfied during the transient regime, e.g. in Figure \ref{fig: trajectory comparison LTV DC motor}, the system still converges to the origin at time $t\approx T$ and displays significant tolerance to disturbances (see Figure \ref{fig: LTV DC motor state comparison}).

%% file: 06_Conclusion.tex
\section{Concluding Remarks}\label{sec: conclusion}

In this paper, the concept of dual modulating functions and dual modulations is introduced, showing that the modulating function method is not only an estimation framework, but also a control theory framework for LTV MIMO systems of arbitrary order. In addition, necessary and sufficient conditions for the existence of the proposed control law were obtained, together with the relation between the proposed approach and well-known asymptotic controllers, in particular state feedback, output feedback, and LTI sliding mode control. In addition, necessary and sufficient conditions were also obtained for finite-time stabilizing controllers, and the self-modulation operator and controllability bracket were defined and shown to be practical tools to evaluate the existence of the proposed control laws for controllable systems, with the reachability gramian itself being obtained as a particular case. Moreover, a fixed-time control law for LTV MIMO systems was proposed by combining explicit selection of the modulating functions with a dual auxiliary system.

The proposed fixed-time controller was then illustrated with numerical simulations, displaying performance comparable to LQR algorithms, and even having a lower actuation use than the benchmark LQR in one of the examples. Furthermore, the proposed method displayed reasonable robustness to the number of partitions in the numerical integration algorithm and external disturbances.

Additionally, the proposed approach explicitly utilizes the system model, allowing for direct evaluation of parameter changes, and has the same structure for both LTI and LTV system, with the difference being that, while LTI systems allow for pre-computation of the gains, LTV systems require computing them online.

Further work on this topic could involve an implementation of the method on embedded hardware for performance evaluation; continuing the investigation on the impact of the modulating function selection on the transient response and numerical robustness; further investigating the robustness of the approach to external disturbances; extending the method for other control laws; extending the proposed regulator to tracking problems; or studying which families of kernels satisfy the controllability bracket and the finite-time stabilization conditions.